\documentclass{article}
\usepackage{amsthm, amssymb, amsmath}
\usepackage{enumerate}
\usepackage{amsthm}
\usepackage{graphicx}
\numberwithin{equation}{section}
\newtheorem{theorem}{Theorem}[section]

\newtheorem{lemma}[theorem]{Lemma}

\newtheorem{remark}{Remark}
\newtheorem{proposition}[theorem]{Proposition}

\huge

\newcommand{\lb}{\left(}
\newcommand{\rb}{\right)}

\newcommand{\nc}{\newcommand}

\nc{\be}{\begin{equation}}
\nc{\la}{\label}
\nc{\ba}{\begin{array}}
\nc{\ea}{\end{array}}
\nc{\bs}{\begin{split}}
\nc{\es}{\end{split}}

\nc{\J}{\mathbb J}

\nc{\pt}{\partial_t}
\nc{\ptt}{\partial_t^2}
\nc{\e}{\epsilon}
\nc{\lam}{\lambda}
\nc{\G}{\Gamma}
\nc{\g}{\gamma}
\nc{\al}{\alpha}
\nc{\del}{\delta}
\nc{\Om}{\Omega}

\newcommand{\ra}{\rightarrow}

\date{}

\newcommand{\DETAILS}[1]{}

\begin{document}
\title{Emission of Cherenkov Radiation as a Mechanism for Hamiltonian Friction}
\author{J\"urg Fr\"ohlich\footnote{juerg@phys.ethz.ch; \text{ }present address: School of Mathematics, IAS, Princeton, NJ 08540, USA} \ \text{   }and \ \text{} Zhou Gang\footnote{gzhou@caltech.edu;\ \text{   } Partially supported by NSF grant DMS-1308985}}\maketitle
\setlength{\leftmargin}{.1in}
\setlength{\rightmargin}{.1in}
\normalsize \vskip.1in
\setcounter{page}{1} \setlength{\leftmargin}{.1in}
\setlength{\rightmargin}{.1in}
\large
\centerline{$^{\ast}$Institute of Theoretical Physics, ETH Zurich, CH-8093 Zurich, Switzerland}
\begin{center}{$^{\dagger}$ Division of Physics, Mathematics and Astronomy,
                     California Institute of Technology, \\ Pasadena, CA 91125, USA}
\end{center}
\date

%\fixNumberingInArticle
\setlength{\leftmargin}{.1in}
\setlength{\rightmargin}{.1in}
%\NowFootNum \fixNumberingInArticle
\normalsize \vskip.1in
\setcounter{page}{1} \setlength{\leftmargin}{.1in}
\setlength{\rightmargin}{.1in}
\large

\section*{Abstract}
We study the motion of a heavy tracer particle weakly coupled to a dense, weakly interacting Bose gas exhibiting Bose-Einstein condensation. In the so-called mean-field limit, the dynamics of this system approaches one determined by nonlinear Hamiltonian evolution equations. We prove that if the initial speed of the tracer particle is above the speed of sound in the Bose gas, and for a suitable class of initial states of the Bose gas, the particle decelerates due to emission of Cherenkov radiation of sound waves, and its motion approaches a uniform motion at the speed of sound, as time $t$ tends to $\infty$.
\tableofcontents

%\tableofcontents
%%%%%%%%%%%%%%%%%%%%%%%%%%%%%%%%%%%%%%%%%%%%%%%%%%%%%%%%%%%%%%%%%%%%%%
%%%%%%%%%%%%%%%%%%%%%%%%%%%%%%%%%%%%%%%%%%%%%%%%%%%%%%%%%%%%%%%%%%%%%%
%%%%%%%%%%%%%%%%%%%%%%%%%%%%%%%%%%%%%%%%%%%%%%%%%%%%%%%%%%%%%%%%%%%%%%
%%%%%%%%%%%%%%%%%%%%%%%%%%%%%%%%%%%%%%%%%%%%%%%%%%%%%%%%%%%%%%%%%%%%%%
\section{Background from Physics and Equations of Motion}
In this paper we study the motion of a very heavy tracer particle coupled to a very dense, very weakly interacting Bose gas at zero temperature exhibiting Bose-Einstein condensation. In an interacting Bose gas at positive density and zero temperature, the speed of sound is strictly positive. If the initial speed of the tracer particle is well below the speed of sound in the gas one expects that the motion of the particle approaches a uniform (inertial) motion at large times. A result in this direction has recently been established in a certain limiting regime (the ``mean-field-Bogolubov limit'') of the Bose gas in \cite{EFGSS}. In the present paper, we prove results complementary to those in \cite{EFGSS} for the same model: Assuming that the initial speed of the tracer particle is larger than the speed of sound in the Bose gas, we show that this particle decelerates by emission of Cherenkov radiation of sound waves into the gas until its speed is equal to (or smaller than) the speed of sound. For some earlier results on related models, see also
 ~\cite{Spohn2004, KM}.

To be specific, we consider a tracer particle of mass $\Lambda M$ coupled through two-body forces of strength $O(1)$ to atoms of mass $m>0$ in a Bose gas of density $\Lambda \rho_{0}/g^{2}$. The gas atoms interact through two-body forces of strength $\Lambda^{-2} \kappa$. The parameters $M$, $\rho_{0}$ and
$\kappa$ are kept fixed, while $\Lambda$ is allowed to vary between 1 and $\infty$, (and the choice of $g$ varies from one model to another, as described below). In the so-called mean-field limit, which corresponds to letting $\Lambda \rightarrow \infty$ (see ~\cite{MR2858064, DFPP}), the dynamics of the system approaches one governed by the following \textit{classical} Hamiltonian equations of motion:
\begin{align}
\dot{X_t}=&\frac{P_t}{M},\quad\quad
\dot{P_t}=-\nabla_{X}\Phi(X_t)+g\int dx \text{ }\nabla_{x}W(X_t-x)\lbrace|\alpha_t(x)|^2
-\frac{\rho_0}{g^2}\rbrace,    \label{XPeqns2}\\
i\dot{\alpha}_t(x)=&\lb-\frac{1}{2m}\Delta+gW(X_t-x)\rb\alpha_t(x) %\nonumber\\
+\kappa \lb{\phi *\lbrace|\alpha_t|^2-\frac{\rho_0}{g^2}\rbrace}\rb(x)\text{ }\alpha_t(x),              \label{alphaeqn}
\end{align}
In Eqs.\eqref{XPeqns2} and \eqref{alphaeqn}, $X_t\in \mathbb{R}^3$ and $P_{t}\in \mathbb{R}^3$ are the position and momentum of the tracer particle at time $t$, respectively, $\Phi$ is the potential of an external force acting on the particle, and $\alpha_t(x)$ is the Ginzburg-Landau order-parameter field describing the state of the Bose gas in the mean-field limit at time $t$. Furthermore, $W$ and $\phi$ are two-body potentials of short range, $\kappa\phi$ is assumed to be of positive type (to ensure stability of the gas against collapse), and $g$ and $\kappa \geq 0$ are coupling constants. The interpretation of $|\alpha_{t}(x)|^2$ is that of the density of bosonic atoms at the point $x$ of physical space $\mathbb{R}^3$,
at time $t$. The global phase of $\alpha_t$ is not an observable quantity.
The symbol $*$ in \eqref{alphaeqn} denotes convolution.

Eqs. \eqref{XPeqns2} and \eqref{alphaeqn} are the Hamiltonian equations of motion corresponding to the following Hamilton functional
\begin{align}\label{Ham}
 H(X,P,\alpha,\bar{\alpha})&=&\frac{P^2}{2M}+\Phi(X)+\int dx\text{ } \{
\frac{1}{2m}|\nabla\alpha(x)|^2+gW(X-x)\ (|\alpha(x)|^2-\frac{\rho_0}{g^2})\ \}\\
& &+\frac{\kappa}{2}\int dx \int dy
\ (|\alpha(y)|^2-\frac{\rho_0}{g^2})\ \phi(y-x)
\ (|\alpha(x)|^2-\frac{\rho_0}{g^2}).\nonumber
\end{align}
The phase space of the system is given by $\mathbb{R}^{6}\times\mathcal{H}$, where $\mathcal{H}$ is a function space defined below. Poisson brackets are defined on phase space by
\begin{equation}\label{PoissonP,X}
\lbrace X^i,X^j\rbrace = \lbrace P_i,P_j \rbrace = 0, \lbrace X^i, P_j \rbrace = \delta^i_j,
\end{equation}
and
\begin{equation}\label{PoissonG-L}
\lbrace \alpha^{\sharp}(x), \alpha^{\sharp}(y)\rbrace = 0, \lbrace \alpha(x), \bar {\alpha}(y) \rbrace = -i \delta (x-y).
\end{equation}

We impose the conditions that $\nabla \alpha$ is square-integrable in $x$ and that $|\alpha|^2-\frac{\rho_0}{g^2}$ is integrable. In the present paper, we also require that
$\alpha(x) \rightarrow \sqrt{\frac{\rho_0}{g^2}}$, as $|x| \rightarrow \infty$. These conditions define
the space $\mathcal{H}$, (which is an affine space of complex-valued functions on $\mathbb{R}^{3}$). The boundary condition at $\infty$ explicitly breaks
invariance under global gauge transformations, $\alpha^{\sharp}(x) \rightarrow e^{\pm i\theta} \alpha^{\sharp}(x)$, where $\theta$ is an arbitrary angle.
Given these boundary conditions, it is natural to define a new function $\beta$ by setting
\begin{equation} \label{alphabetarel}
\alpha(x)=:\sqrt{\frac{\rho_0}{g^2}}+\beta(x),
\end{equation}
with $\beta(x)\rightarrow 0$, as $|x|\rightarrow \infty.$
The equations of motion then read
\begin{align}
\dot{X_t}=&\frac{P_t}{M},\quad\quad
\dot{P_t}=-\nabla_{X}\Phi(X_t)+g\int\nabla_{x}W(X_t-x)\lb|\beta_t(x)|^2
+2\sqrt{\frac{\rho_0}{g^2}}Re\beta_t(x)\rb dx,    \label{XPeqns}\\
i\dot{\beta}_t(x)=&\lb-\frac{1}{2m}\Delta+gW(X_t-x)\rb
\beta_t(x)+\sqrt{\rho_0}W(X_t-x)                  \nonumber\\
+&\kappa\lb\phi *\lb|\beta_t|^2+2\sqrt{\frac{\rho_0}{g^2}}
Re\beta_t\rb\rb(x)\lb\beta_t(x)+\sqrt{\frac{\rho_0}{g^2}}\rb. \label{betaeqn}
\end{align}

The Hamilton functional giving rise to these equations is obtained from \eqref{Ham} by inserting Eq.
\eqref{alphabetarel}. It is easy to see that, under rather weak assumptions on the potentials $W$ and $\phi$, Eqs. \eqref{XPeqns} and \eqref{betaeqn} have static solutions, and that if the external force acting on the tracer particle vanishes ($\Phi \equiv 0$) they have ``traveling wave solutions'', provided the speed of the particle is smaller than or equal to the speed of sound in the Bose gas; see \cite{FGSS, EFGSS}. These solutions correspond to an inertial motion of the tracer particle at a constant velocity, with the particle accompanied by a ``splash'' in the Bose gas. (Quantum mechanically, this splash corresponds to a coherent state of gas atoms
and causes decoherence in particle-position space, which allows for an essentially ``classical" detection of the particle trajectory.)
If, initially, the speed of the tracer particle is larger than the speed of sound it emits sound waves into the condensate (Cherenkov radiation), which causes $friction$. As a consequence, the particle loses kinetic energy until its speed has dropped to the speed of sound in the Bose gas (or below). This phenomenon has been analyzed for a simple model (the B-model defined below) in \cite{MR2858064}. Cherenkov radiation in a more subtle model (the E-model) is described in the present paper.

We remark that, originally, \textit{`Cherenkov radiation'} has been the name for the phenomenon (observed, e.g., in nuclear reactors) that charged particles (electrons) moving through an optically dense medium (water) at a speed larger than the speed of light in the medium emit electromagnetic radiation until their speed has dropped to the speed of light (or below). This phenomenon is described in any good text book on classical electromagnetism; see, e.g., \cite{Jackson}. We believe that a mathematical treatment could be accomplished along the lines of the analysis presented in this paper.

The following models are of interest (see \cite{FGSS, MR2858064}):
\begin{itemize}
\item[B]-Model: $\kappa=0$ (ideal Bose gas), $g\rightarrow 0$; see \cite{MR2858064}.
\item[C]-Model: $\kappa=0$, but $g\neq 0;$ see ~\cite{EG2012}.
\item[E]-Model: $2\kappa\rho_0/g^2 :=\lambda=$const., with $g, \kappa\ra  0$ (``Bogolubov limit'');
this paper.
\item[G]-Model: $\kappa>0$ and $g\neq 0$.
\end{itemize}

In this paper, we focus our attention on a special case of the E-model, with $\Phi\equiv 0$ and $\phi(x) = \delta(x)$. (The G-Model is presently under study.) The equations of motion then take the form
\begin{align}
\dot{X}_{t}=&\frac{1}{M}P_{t}, \ \dot{P}_{t}=\sqrt{\rho_0}Re\langle \nabla_{x}W^{X_{t}},\ \beta_{t}\rangle \label{eq:momentum}\\
i\dot{\beta}_{t}=&-\frac{1}{2m}\Delta \beta_t+\lambda Re\beta_{t}+\sqrt{\rho_0}W^{X_t}, \label{eq:field}
\end{align}
where
\begin{align}
W^{X}(x):=W(X-x).
\end{align}
The speed of sound is given by $\sqrt{\frac{\lambda}{2m}}$ by the fact that the linear equation for $\beta_t$ behaves like a wave equation in a small neighborhood of zero momentum, and $\sqrt{\frac{\lambda}{2m}}$ is the speed of propagation.

The Hamilton functional giving rise to these equations of motion is found to be
\begin{align}\label{eq:conserv}
H(X, P;\ \bar\beta,\beta):=\frac{|P|^{2}}{2M}+\frac{1}{2m}\int_{\mathbb{R}^3} |\nabla \beta|^2 dx+\lambda \int_{\mathbb{R}^3} |Re\beta|^2 dx+2\sqrt{\rho_0}\int_{\mathbb{R}^3}W^{X}Re \beta dx.
\end{align}

In the present paper we consider the supersonic regime, namely the initial speed $|\frac{P_0}{M}|$ of the tracer particle is larger than the speed of sound, $v_{s}$. For the subsonic regime we refer to our previous papers \cite{EFGSS, FG2013}. In contrast to sonic and subsonic particle motions, inertial supersonic particle motions do $not$ exist, i.e., the equations of motion do not have traveling wave solutions propagating at $supersonic$ speeds, as shown in \cite{EFGSS}.
We propose to construct solutions of the equations of motion corresponding to initial conditions at time $t=0$ with the properties that $\beta_0$ is ``small'' in a suitable sense, i.e., the state of the Bose gas is close (or equal) to
the ground state $\beta_0=0$,
and the initial speed of the tracer particle is above the speed of sound,
$v_{s} = \sqrt{\frac{\lambda}{2m}}$. Assuming that the interaction potential $\sqrt{\rho_0}W$ is sufficiently weak, specifically that $\rho_0$ is sufficiently small and $W=(-\Delta)^{n}V$, where $V$ is of rapid decay and smooth and $n\geq \frac{3}{4}$, we prove that, for such initial conditions, the particle motion approaches an inertial one at the speed of sound,
as time $t \rightarrow \infty$.

The effective equation of motion governing $P_t$ has the form
\begin{align*}
\dot{P}_t=-\rho_0 F(|P_t|) \frac{P_t}{|P_t|}+\rho_0 \cdot Remainder(t),
\end{align*}
where $F(|P|)$ is a scalar function given by
$$F(|P|):=\Lambda(\frac{|P|}{M}-\sqrt{\frac{\lambda}{2m}})
\big[\frac{|P|}{M}-\sqrt{\frac{\lambda}{2m}}\big]^{2n+2},$$
$\Lambda$ is a real-valued function satisfying
$\Lambda(x)\geq C_0>0$, for $x>0$, with $C_0>0$ some constant, and $\Lambda(x)\equiv 0$, for $x\leq 0$; (see Lemma \ref{LM:est1} below), and $n$ is the exponent in $W=(-\Delta)^{n}V$. The negative sign in front of $F(|P_t|)$ on the right side of the equation of motion for $P_{t}$ implies that this term describes a friction force, whose direction is opposite to the direction of $P_{t}$. This friction force results from the instability of supersonic inertial motion under turning on the interaction, $W$, between the tracer particle and the Bose gas. It has the form expected from formal `Fermi-Golden-Rule' type calculations; (see also \cite{MR2858064}).

The central observation made in the present paper is that, for $n\geq \frac{3}{4}$, and for sufficiently large times; (more specifically for $t\geq \rho_0^{-\frac{1}{10}})$
\begin{align}
\frac{|Remainder(t)|}{F(|P_t|)}\leq \rho_0^{\frac{1}{10}}.
\end{align}
Recall that $\rho_0$ is a small constant.
This implies that the equation of motion for $P_{t}$ is effectively governed by the friction force $-\rho_0 F(|P_t|) \frac{P_t}{|P_t|}$, and this will imply our main result concerning the asymptotic behavior of the particle motion.

Our paper is organized as follows: In Section \ref{sec:MainTHM}, we describe our main result -- Theorem \ref{THM:main} -- which has two parts, the first one concerning the motion of the particle, and the second one concerning the state, $\beta_t$, of the Bose gas. Part (1) is proven in Section \ref{sec:D123}, Part (2) in Section \ref{sec:formState}. Numerous technical problems that come up in the proofs are solved in subsequent sections and in appendices.

\textit{Notations}: \text{  }By $\mathcal{H}^{k}, \ k=1,2,3,\cdots,$ we  denote the Sobolev spaces of complex-valued functions on $\mathbb{R}^{3}$ equipped with the norms
$$\|f\|_{\mathcal{H}^{k}}:=\|(1-\Delta)^{\frac{k}{2}}f\|_2.$$
For positive quantities $a$ and $b$, the meaning of ``$a\lesssim b$'' (or ``$a\gtrsim b$") is that there exists a positive constant $C$ such that $a\leq C b$ ($a\geq C b$, respectively).
The scalar product of two square-integrable functions, $f$ and $g$, on $\mathbb{R}^{3}$ is given by
 $$ \langle f,\ g\rangle:=\int f(x) \bar{g}(x)\ dx.$$
\section*{Acknowledgements}
We are indebted to Daniel Egli, Burak Erdogan, Eduard Kirr, Nikolas Tzirakis and, especially, to Israel Michael Sigal for numerous very illuminating discussions.

Our collaboration on the problems solved in this paper has been made possible by a stay at the Institute for Advanced Study in Princeton. We wish to thank our colleagues at the School of Mathematics, in particular Thomas C. Spencer, and the staff of the Institute for hospitality. The stay of J. F. at IAS has been supported by 'The Fund for Math' and `The Robert and Luisa Fernholz Visiting Professorship Fund'. The research of Zhou Gang was partly supported by NSF grant DMS-1308985.

\section{Statement of the Main Result and Strategy of Proof}\label{sec:MainTHM}
In this section we describe our hypotheses on the potential $W$ and on the choice of initial conditions;
(see hypotheses $(A)$ and $(B)$, below). We then state our main results -- see Theorem \ref{THM:main} -- and present an outline of the strategy of our proof.

In what follows we assume that $W$ is of the form
\begin{align}\label{eq:addFriction}
W=(-\Delta)^{n}V,
\end{align}
 for some $n\geq \frac{3}{4}$,
where $V$ is a smooth, spherically symmetric, real function that decays exponentially fast at spatial $\infty$ and has the property that $\hat{V}(0)\not=0.$
\begin{remark}
The choice of a sufficiently large value of $n$ (and of an appropriate initial condition, $\beta_{0}$, for the Bose gas) will be important in our derivation of the following two features of particle motion that play a key role in our analysis: (i) the magnitude of the momentum of the particle tends to decrease in time; and (ii) the direction of motion of the  particle is close to constant. (These features will appear as conditions (I) and (II) in Lemma \ref{LM:est}, below, which will be used to obtain the crucial upper bound in Eq. \eqref{eq:lowerBG}.) Let us attempt to explain the connection between the value of $n$ and features (i) and (ii) in a heuristic way.

In \eqref{eq:D3} we will find a differential equation for $P_t$ of the form
\begin{align*}
\dot{P}_t=D_1(P_t)+Remainder(t),
\end{align*}
where the vector-valued function $D_1$ on $ \mathbb{R}^3$ is given by
$$D_1(P_t)=-\rho_0\Lambda(\frac{|P_t|}{M}-\sqrt{\frac{\lambda}{2m}})[\frac{|P_t|}{M}-\sqrt{\frac{\lambda}{2m}}]^{2n+2} \frac{P_t}{|P_t|},$$
($|D_{1}(P)| = F(|P|)$), with $\Lambda$ a real-valued function satisfying
$\Lambda(x)>C_0>0$, for $x>0$, and $\Lambda(x)\equiv 0$, for $x\leq 0$; (see Lemma \ref{LM:est1} below).

We observe that, in order to derive features (i) and (ii), above, it suffices to show that
\begin{align}\label{eq:RemD1}
\frac{|Remainder(t)|}{|D_1(P_t)|}\leq \rho_0^{\frac{1}{10}},
\end{align} which is small,
for large enough times, $t$, since any solution to the simpler equation
\begin{align}\label{eq:tildeP}
\dot{P}_t=D_1(P_t),\text{  }\ \text{with}\  \frac{|P_0|}{M}>\sqrt{\frac{\lambda}{2m}},
\end{align} exhibits features (i) and (ii).

The key observation is that, in order to ensure that \eqref{eq:RemD1} holds, it suffices to make
$|D_1(P_t)|$ decrease sufficiently slowly as a function of time $t$. This property can be shown to hold, provided the exponent $n$ is chosen large enough. On a heuristic level, this is seen as follows: A solution to Eq. ~\eqref{eq:tildeP} with $\frac{\vert P_0 \vert}{M} > v_{s} = \sqrt{\frac{\lambda}{2m}}$ obeys upper and lower bounds
\begin{align*}
C_1(1+\rho_0 t)^{-\frac{1}{2n+1}}\leq \frac{|P_t|}{M}-\sqrt{\frac{\lambda}{2m}}\leq C_2(1+\rho_0 t)^{-\frac{1}{2n+1}},
\end{align*}
for some positive constants $C_1$ and $\ C_2$.
Hence there is a constant $C_3>0$ such that
 $$|D_1(P_t)|\geq C_3 \rho_0 (1+\rho_0 t)^{-1-\frac{1}{2n+1}},$$
 which obviously decays more slowly in $t$ the larger the exponent $n$ is, and this turns out to imply (\ref{eq:RemD1}).
We will see that our assumption that $n\geq \frac{3}{4}$ suffices to make the arguments just sketched mathematically precise.

Physically, we choose $n$ large enough to soften the friction between the tracer particle and the Bose gas. As shown above, the large $n$ is, the slower $|\dot{P}_t|$ decays.
\end{remark}

Before we are able to formulate the \textit{main result} established in this paper we have to state a second important assumption required in our analysis, namely a so-called \textit{Fermi-Golden-Rule} condition: Since $V$ is spherically symmetric, so is its Fourier transformation, $\hat{V}$. We assume that \begin{equation}\label{eq:FGR}
 \hat{V}(k)=\hat{V}(|k|)=0, \text{only for a discrete set of values of}  \text{  } \vert k\vert.
\end{equation}
This assumption enables us to show that $D_1(P)$ does not vanish and, with the condition on the exponent $n$, to derive the form of $D_1(P)$ required in our analysis.

A typical example of such a potential is the Gaussian, $V(x)=e^{-|x|^2}.$

We are now ready to state our main result.

\begin{theorem}\label{THM:main}
Suppose the potential $W$ is of the form \eqref{eq:addFriction}, where $n \geq \frac{3}{4}$ and $V$ is smooth, spherically symmetric and of exponential decay at $\infty$, with $\hat{V}(0)\not=0$, and suppose the Fermi Golden Rule condition \eqref{eq:FGR} holds. Suppose, furthermore, that $M$, $m$ and $\lambda$ are all of order $O(1)$, and that the parameter $\rho_0>0$ proportional to the density of the gas is sufficiently small. We also assume that
\begin{itemize}
\item[(A)] initially, the state, $\beta_0$, of the Bose gas is close to (or equal to) the ground state
 $\beta=0$; more specifically $\|e^{\epsilon_0 |x|} \beta_0\|_{\mathcal{H}^3}\leq  \sqrt{\rho_0}$, for some $\epsilon_0>0$; and
\item[(B)]
the initial speed of the tracer particle is larger than the speed of sound,
$v_{s} = \sqrt{\frac{\lambda}{2m}} $, in the Bose gas; more specifically $|v_0|=\frac{1}{M}|P_0|$ satisfies $ \frac{11}{10}\sqrt{\frac{\lambda}{2m}}\leq |v_0|\leq 10\sqrt{\frac{\lambda}{2m}}$.
\end{itemize}
Then the following results hold true:
\begin{itemize}
\item[(1)] At large times, the motion of the tracer particle approaches a uniform (inertial) sonic  motion: There exists some $P_{\infty}\in \mathbb{R}^3$, with $\frac{1}{M}|P_{\infty}|= \sqrt{\frac{\lambda}{2m}}$, such that $P_{t}\rightarrow P_{\infty}$, as $t\rightarrow \infty.$
    The momentum $P_t$ is the solution of an equation of the form
\begin{align}\label{Pdot}
    \dot{P}_t=-\rho_0 \Lambda (\frac{|P_{t}|}{M}) [\frac{|P_t|}{M}-\sqrt{\frac{\lambda}{2m}}]^{2+2n}\frac{P_t}{|P_t|}+\text{Remainder}(t)
    \end{align}
where the exponent $n \geq \frac{3}{4}$ is as in \eqref{eq:addFriction}, $\Lambda(\frac{|P_t|}{M})\geq C_0$, for some strictly positive constant $C_0$, provided
    $\frac{|P_t|}{M}>\sqrt{\frac{\lambda}{2m}}$, and $\equiv 0$, otherwise; and
    $$|\frac{\text{Remainder(t)}}{\rho_0 \Lambda (\frac{|P_{t}|}{M}) [\frac{|P_t|}{M}-\sqrt{\frac{\lambda}{2m}}]^{2+2n}}|\leq \rho_0^{\frac{1}{10}},$$ for large times, specifically for $t\geq \rho_{0}^{-\frac{1}{10}}.$

\item[(2)] The function $\beta_t$ describing the state of the Bose gas approaches a ``traveling wave'' accompanying the particle, in the sense that there exists a function $\beta_{\infty}\in L^{2}(\mathbb{R}^3)$ with the property that
\begin{align}\label{eq:betaconv}
\|\beta_t-\beta_{\infty}(\cdot - X_{t})\|_{\infty}\rightarrow 0,\ \text{as}\ t\rightarrow \infty,
\end{align}
where $X_{t} = X_{0} + \int_{0}^{t} \frac {P_s}{M}\ ds$ is the position of the particle at time $t$.
\end{itemize}
\end{theorem}

Statement (1) of Theorem \ref{THM:main} is proven in Section \ref{sec:D123}, Statement (2) in Section \ref{sec:formState}. Various auxiliary results are stated and proven in later sections and some appendices.

As a preliminary result needed in the proof of Theorem \ref{THM:main}, one must establish local and global well-posedness of the equations of motion \eqref{eq:momentum} and \eqref{eq:field}. This is quite easily accomplished, because, for any given particle trajectory
$\lbrace X_t \rbrace_{0 \leq t < \infty}$, the equation for $\beta_t$, namely Eq. \eqref{eq:field}, is linear.
A detailed proof of well-posedness has been presented in \cite{EFGSS} and applies to the model studied in the present paper.

In the remainder of this section we present the main ideas used in the proof of Statement (1) of Theorem \ref{THM:main}. The proof of Statement (2) turns out to be significantly easier than the proof of (1), given that (1) holds true. We will therefore not sketch it here.

In order not to clutter our arguments with clumsy formulae, we rescale dimensionful variables such that
\begin{align}
2m=M=\lambda=1,\ \text{and}\ |\hat{V}(0)|=1.
\end{align}
We will assume $\rho_0>0$ to be sufficiently small wherever needed.

To begin with, we recast the equations of motion \eqref{eq:momentum} and \eqref{eq:field} in a more convenient form. Note that equation \eqref{eq:field} is merely real-linear, rather than complex-linear, in $\beta_t$. It is therefore convenient to rewrite it as a system of equations for
$\text{Re}\beta_t$ and $\text{Im}\beta_t$.
We thus introduce a vector function, ${\bf{h}}_{t}:\ \mathbb{R}^3\rightarrow \mathbb{R}^2$, by setting
\begin{align}\label{eq:vForm}
{\bf{h}}_{t}  (x-X_{t})=\left[
\begin{array}{ccc}
\text{Re}\beta_{t}(x)\\
\text{Im}\beta_{t}(x)
\end{array}
\right].
\end{align}
Then Eqs. \eqref{eq:momentum} and \eqref{eq:field} become
\begin{align}
\dot{X}_{t}=&P_{t}, \ \dot{P}_{t}=\sqrt{\rho_0}\langle
\left[
\begin{array}{ccc}
\nabla_{x}W\\
0
\end{array}
\right],
 \bf{h}_{t}\rangle,\label{eq:traj1}\\
\dot{\bf{h}}_{t}=&H(t) {\bf{h}}_{t}-\sqrt{\rho_0} \left[
\begin{array}{ccc}
0\\
W
\end{array}
\right], \label{eq:field1}
\end{align}
where $H(t)$ is the $2\times 2$ matrix operator given by
\begin{align}
H(t):=\left[
\begin{array}{ccc}
P_{t}\cdot \nabla_{x}& -\Delta\\
-(-\Delta+1) & P_{t}\cdot\nabla_{x}
\end{array}
\right].
\end{align}
To study these equations it is convenient to introduce a new vector function, $\bf{\delta}_{t}$, by setting
\begin{align}
{\bf{h}}_{t}=\sqrt{\rho_0}[H(t)-0]^{-1} \left[
\begin{array}{ccc}
0\\
W
\end{array}
\right]+\bf{\delta}_t.\label{eq:decom1}
\end{align} Here the function $[H(t)-0]^{-1} \left[
\begin{array}{ccc}
0\\
W
\end{array}
\right]$ should be understood as
\begin{align*}
[H(t)-0]^{-1} \left[
\begin{array}{ccc}
0\\
W
\end{array}
\right]:=\lim_{\epsilon\rightarrow 0^{+}}[H(t)-\epsilon]^{-1} \left[
\begin{array}{ccc}
0\\
W
\end{array}
\right].
\end{align*}
Rewriting Eqs. \eqref{eq:traj1} and \eqref{eq:field1} in terms of $\bf{\delta}_t$, we find the equations
\begin{align}
\dot{X}_{t}=&P_{t}, \\
\dot{P}_{t}=&\rho_0 \langle\left[
\begin{array}{ccc}
\nabla_{x}W\\
0
\end{array}
\right], [H(t)-0]^{-1} \left[
\begin{array}{ccc}
0\\
W
\end{array}
\right]\rangle+\sqrt{\rho_0}\langle
\left[
\begin{array}{ccc}
\nabla_{x}W\\
0
\end{array}
\right],
 \bf{\delta}_{t}\rangle\label{eq:traj2},\\
\dot{\bf{\delta}}_{t}=&H(t) {\bf{\delta}}_{t}+\sqrt{\rho_0} [H(t)-0]^{-2} \dot{P_t}\cdot \nabla_{x}\left[
\begin{array}{ccc}
0\\
W
\end{array}
\right].\label{eq:field2}
\end{align}
The initial condition for $\delta_t$ is given by
\begin{align}\label{eq:delta0}
\delta_0=-\sqrt{\rho_0} [H(0)-0]^{-1}\left[
\begin{array}{ccc}
0\\
W
\end{array}
\right]+\bf{h}_0.
\end{align}

Applying Duhamel's principle to equation \eqref{eq:field2} we obtain the integral equation
\begin{align}
{\bf{\delta}_{t}}=&U(t,0){\bf{\delta}_0}  +\sqrt{\rho_0} \int_{0}^{t} ds \text{  }U(t,s)  [H(s)-0]^{-2} \dot{P_s}\cdot \nabla_{x}\   \left[
\begin{array}{ccc}
0\\
W
\end{array}
\right]\nonumber\\
=&U(t,0) {\bf{h}_0} -\sqrt{\rho_0}U(t,0) [H(0)-0]^{-1}\left[
\begin{array}{ccc}
0\\
W
\end{array}\right]\nonumber\\
+&\sqrt{\rho_0} \int_{0}^{t} ds \text{   }U(t,s)  [H(s)-0]^{-2} \dot{P_s}\cdot \nabla_{x}\   \left[
\begin{array}{ccc}
0\\
W
\end{array}\right], \label{eq:expreDelta}
\end{align}
where $U(t,s)$ is the propagator (from time $s$ to time $t$) generated by the time-dependent operator $H(\cdot)$,  and we have used the expression for $\delta_0$ in \eqref{eq:delta0}.

Inserting Eq. \eqref{eq:expreDelta} into the equation of motion \eqref{eq:traj2} for $P_t$ we obtain an effective equation of motion for the tracer particle
\begin{align}
\dot{P}_t
=&\rho_0 \langle\left[
\begin{array}{ccc}
\nabla_{x}W\\
0
\end{array}
\right], [H(t)-0]^{-1} \left[
\begin{array}{ccc}
0\\
W
\end{array}
\right]\rangle\\
&+\sqrt{\rho_0}\langle
\left[
\begin{array}{ccc}
\nabla_{x}W\\
0
\end{array}
\right],
 U(t,0) {\bf{h}_0}\rangle\nonumber\\
&-\rho_0 \langle
\left[
\begin{array}{ccc}
\nabla_{x}W\\
0
\end{array}
\right],
U(t,0) [H(0)-0]^{-1}\left[
\begin{array}{ccc}
0\\
W
\end{array}
 \right]\rangle
 \nonumber\\
&+\rho_0\langle
\left[
\begin{array}{ccc}
\nabla_{x}W\\
0
\end{array}
\right],\ \int_{0}^{t} ds \text{   }U(t,s) [H(s)-0]^{-2} \dot{P_s}\cdot \nabla_{x}\   \left[
\begin{array}{ccc}
0\\
W
\end{array}
\right]
\rangle\nonumber\\
=:& D_1(P_t)+D_2(P,t)+D_3(P,t)+D_4(P,t),\label{eq:D3}
\end{align}
where $D_{k}$ corresponds to the term on the $k^{th}$ line, above, for $k=1,2,3,4$, respectively, and
$P$ stands for $\lbrace P_s\vert 0\leq s \leq t\rbrace$.

The term $D_1$ is the most important one, and we have detailed information on its behavior.
\begin{lemma}\label{LM:est1}
Suppose that the Fermi-Golden-Rule condition in \eqref{eq:FGR} holds.
Then the term $D_1$ is given by
\begin{align}\label{eq:D1}
D_1(P_t)=-\rho_0 \Lambda(|P_t|) \Bigg[|P_t|-1\Bigg]^{2+2n} \frac{P_t}{|P_t|},
\end{align}
where $\Lambda$ is a function that is discontinuous at $1$ and  is smooth elsewhere. Most importantly, there exists a constant $C_0>0$ such that if $x>1$ then $\Lambda(x)\geq C_0$, and if $x\leq 1$ then $\Lambda(x)\equiv 0.$
\end{lemma}
This lemma will be proven in Section \ref{sec:LemEst}.

Next, we describe our strategy to control the behavior of the particle momentum $P_{t}$, for large times $t$.

To start with, we remark that equations somewhat similar to \eqref{eq:D3} have been studied in
 ~\cite{MR1681113,TsaiYau02,BuSu,GaSi2007, MR2182081}, where the motion of solitary waves (ground states) of nonlinear Schr\"odinger equations has been studied. We adapt the main ideas developed in these papers to the present context.

Using the equation of motion \eqref{eq:D3}, we find that
\begin{align}
\frac{1}{2} \frac{d}{dt} |P_t|^2=D_1(P_t)\cdot P_t+\sum_{k=2}^{4}D_k(P,t) \cdot P_t
=& -|D_1(P_t)||P_t|+\sum_{k=2}^{4}D_k(P,t) \cdot P_t\nonumber
\end{align} where the identity $D_1=-|D_1|\frac{P_t}{|P_t|}$ in \eqref{eq:D1} has been used. Dividing both sides by $|P_t|$ we obtain
\begin{align}\label{eq:lower}
\frac{d}{dt} |P_t|=-|D_1(P_t)| +\sum_{k=2}^{4}D_k(P,t) \cdot \frac{P_t}{|P_t|}.
\end{align}

Ideally, the Fermi-Golden-Rule (FGR) term $D_{1}$ will be seen to dominate over the terms $D_2, D_3$ and $D_4$, in the sense that
\begin{align}\label{eq:ideal}
|D_1(P_t)|\geq 2|D_2(P,t)+D_3(P,t)+D_4(P,t)|,
\end{align}
for $t$ large enough. The use of inequality \eqref{eq:ideal} is that it allows us to treat $D_2+D_3+D_4$ as a perturbation of $D_1$. Critically to our analysis, this will permit us to prove upper and lower bounds on $|P_t|$.

We temporarily assume that \eqref{eq:ideal} holds for all times. This assumption, together with \eqref{eq:lower}, implies that
\begin{align}\label{eq:lower2}
-\frac{3C_0}{2}\rho_0 [|P_t|-1]^{2+2n} \leq \frac{d}{dt} [|P_t|-1]\leq -\frac{C_0}{2}\rho_0 [|P_t|-1]^{2+2n}
\end{align}
and, dividing both sides by $[|P_t|-1]^{-2-2n}$, one concludes that
\begin{align}
\frac{C_0}{2}[1+2n]\rho_0  \leq \frac{d}{dt}[|P_t|-1]^{-1-2n} \leq \frac{3C_0}{2}[1+2n]\rho_0
\end{align}
This differential inequality yields lower and upper bounds on $|P_t|$, viz.
\begin{align}\label{eq:lower3}
\big[\frac{1}{[|P_0|-1]^{1+2n}}+\rho_0 \Psi t\big]^{-\frac{1}{1+2n}}\geq |P_t|-1\geq [\frac{1}{[|P_0|-1]^{1+2n}}+3\rho_0 \Psi t]^{-\frac{1}{1+2n}},
\end{align}
where $\Psi:=(1+2n)\frac{ C_0 }{2}>0$. Recalling that we have assumed that $|P_0|\geq \frac{11}{10}$, we then find that
\begin{align}
\lim_{t\rightarrow \infty}|P_{t}|=1\text{    }\text{and}\
||P_{t}|-1|\lesssim (1+\rho_0 t)^{-\frac{1}{1+2n}}.\label{eq:unity}
\end{align}
Combining this result with our
equation for $\dot{P}_t$, and using \eqref{eq:ideal}, we find that
\begin{align}
|\dot{P}_t| \lesssim \rho_0(1+\rho_0 t)^{-\frac{2n+2}{2n+1}},
\end{align} the integrability of the right hand side and \eqref{eq:unity} imply that $P_t$
converges to some unit vector $P_{\infty}\in \mathbb{R}^{3}$. This proves Statement (1) of Theorem \ref{THM:main}.

Moreover,
\begin{align}\label{eq:D1Order}
C_1 \rho_0(1+\rho_0 t)^{-1-\frac{1}{1+2n}}\leq |D_1(P_t)|\leq C_2 \rho_0 (1+\rho_0 t)^{-1-\frac{1}{1+2n}},
\end{align}
for some constants $C_1, \ C_2>0$.

\begin{remark}\label{re:twoPoints}
Two points should be stressed:
\begin{itemize}
\item[(a)]
We choose the exponent $n$ in $W=(-\Delta)^n V$ large enough so as to weaken the friction between the tracer particle and the Bose gas. Then the convergence of the momentum $P_t$ to $P_{\infty}$, as
$t\rightarrow \infty$, is quite slow, and this is useful in attempting to prove inequality (\ref{eq:ideal}).
\item[(b)] A slow approach of $P_t$ to $P_{\infty}$ plays a desirable role in that it makes the FGR term,
$D_1(P_t)$, dominate over the other terms in the equation of motion for $P_t$. Then the convergence of $|P_t|$ to 1 (the speed of sound in our units) is seen to be a robust conclusion, which renders the effects of the small terms, $D_k, k=2,3,4$, in the equation of motion for $P_t$ negligible. A large value of the exponent $n$ also plays a favorable role in the proof of the key Proposition \ref{prop:D123}, below.
\end{itemize}
Hence, the larger the value of $n$, the more room one has to maneuver in deriving the behavior of solutions of the equation of motion, Eq. (\ref{eq:D3}), for large times $t$. (See also Remark \ref{re:enoughZero} below.)
\end{remark}

All the arguments presented above depend on the crucial inequality \eqref{eq:ideal}. In what follows we discuss some of the difficulties that are encountered in its proof and some of the ideas used to overcome them.

We first observe that inequality \eqref{eq:ideal} is not necessarily true for small times, e.g., $t=1$. Indeed, every term $D_{k}(P,t=1),\ k=1,2,3,4$, is of order $O(1)$. Hence it is difficult to determine which term dominates over the other ones.

This difficulty is easily circumvented:
We divide the time interval $[0,\infty)$ into two subintervals $[0,\rho_0^{-\frac{1}{10}}]$ and $(\rho_0^{-\frac{1}{10}},\infty)$ and use different arguments to estimate $P_t$ on these intervals.

For $t\in [0,\rho_0^{-\frac{1}{10}}]$, we apply Duhamel's principle to equation \eqref{eq:field1} for
${\bf{h}}_t$ to obtain
\begin{align*}
{\bf{h}}_t=U(t,0) {\bf{h}}_0-\sqrt{\rho_0}\int_{0}^{t} U(t,s) ds\  \left[
\begin{array}{ccc}
0\\
W
\end{array}
\right],
\end{align*}
where $U(t,s)$ is the propagator (from time $s$ to time $t$) generated by the time-dependent operator $H(\cdot)$. Plugging this expression into equation \eqref{eq:traj1} for $\dot{P}$ we get that
\begin{align}
\dot{P}_t=\sqrt{\rho_0} \langle
\left[
\begin{array}{ccc}
\nabla_{x}W\\
0
\end{array}
\right],\
 U(t,0){\bf{h}_{0}}\rangle-\rho_0 \langle
\left[
\begin{array}{ccc}
\nabla_{x}W\\
0
\end{array}
\right],
 \ \int_{0}^{t} U(t,s)\ ds \left[
\begin{array}{ccc}
0\\
W
\end{array}
\right] \rangle. \label{eq:duhPt}
\end{align}
In estimating the terms on the right side of this equation, we use that the propagator $U(t,s)$, $t\geq s\geq 0$, is oscillatory in momentum space. This will yield some decay estimates, as stated in the following lemma.

\begin{lemma}\label{LM:finiteT} We assume that $ {\bf{h_0}} = O(\rho_0^{\frac{1}{2}})$, (see Assumption (A) of
 Theorem \ref{THM:main}).
Then
\begin{align*}
|\langle
\left[
\begin{array}{ccc}
\nabla_{x}W\\
0
\end{array}
\right],\
 U(t,0)\text{  }{\bf{h}_{0}}\rangle|\lesssim &(1+t)^{-\frac{7}{6}} \|(1+ x^2)^2 {\bf{h_0}}\|_{\mathcal{H}^2}\lesssim (1+t)^{-\frac{7}{6}}\rho_0^{\frac{1}{2}}.\\
|\langle
\left[
\begin{array}{ccc}
\nabla_{x}W\\
0
\end{array}
\right],
 \ U(t,s) \left[
\begin{array}{ccc}
0\\
W
\end{array}
\right] \rangle |\lesssim &(1+t-s)^{-\frac{7}{6}},
\end{align*}
for arbitrary times $t$ and $s$.
\end{lemma}

A result in ~\cite{MR2559713} (based on some use of Besov spaces) can be applied to prove this lemma. Later in this paper we will cope with some related, but harder problems, (specifically, with the proofs of the bound \eqref{eq:estKts} and of Lemma \ref{LM:d2d3}, below; our techniques are better adapted to the situation encountered in the present work). Thus, at this point, we omit the details of the proof of Lemma \ref{LM:finiteT}.

Applying the bounds in Lemma \ref{LM:finiteT} to Eq. \eqref{eq:duhPt} and using the smallness of
${\bf{h}}_0$ (see hypothesis (A) of Theorem \ref{THM:main}), we find that
\begin{align}
|\dot{P}_t|\lesssim& \rho_0(1+t)^{-\frac{7}{6}} +\rho_0 \int_{0}^{t}(1+t-s)^{-\frac{7}{6}}\ ds\lesssim \rho_0.
\end{align}

This obviously implies the following proposition.

\begin{proposition}\label{prop:smallVar}
For any $t\in [0,\rho_0^{-\frac{1}{10}}]$,
\begin{align}
|P_t-P_0|\lesssim \rho_0^{\frac{9}{10}}\ \text{and}\ |\dot{P}_t|\lesssim \rho_0.
\end{align}
\end{proposition}

Next we study the behavior of $P_t$ on the time interval $[\rho_0^{-\frac{1}{10}},\infty)$. On this
interval we establish the ``ideal'' inequality \eqref{eq:ideal}, (i.e., the fact that the FGR term dominates over the other three terms). To prove this result, we will use that the propagator $U(t,s)$ is oscillatory in momentum space, and this will yield the necessary smallness.

At the technical level, the following proposition is the most important result in our paper.

\begin{proposition}\label{prop:D123}
The terms $D_2, D_3$ and $D_4$ obey the bounds
\begin{align}
|D_2(P,t)|,\ |D_3(P,t)|\lesssim &\rho_0(1+ t)^{-\frac{3}{2}},\\
|D_{4}(P,t)|\lesssim &\rho_0\int_{0}^{t}(1+t-s)^{-\frac{3}{2}} |\dot{P}_s|\ ds.\label{eq:newEstD4}
\end{align}
\end{proposition}
The proof of this proposition will be presented in Section \ref{sec:D123}; (see, in particular,
Lemmas \ref{LM:est} and \ref{LM:d2d3}).

Heuristically, this proposition and the lower bound on $|D_1(P_t)|$ in \eqref{eq:D1Order} imply the
``ideal'' inequality \eqref{eq:ideal}, for $t \in [\rho_0^{-\frac{1}{10}},\infty)$, and hence Statement (1) of Theorem \ref{THM:main}. To render our arguments mathematically rigorous, which we will accomplish in Section \ref{sec:D123}, some bootstrap argument will be needed.

Next, we present some key elements in the proof of Proposition \ref{prop:D123}. We choose to study the term $D_4(P,t)$, because it is the most involved one (due to the presence of a singularity in
$[H(s)-0]^{-2}$; see Eq. \eqref{eq:D3}).

To get our argument under way, we must express the term $D_4(P,t)$ in a convenient form:
\begin{align}\label{eq:defD4}
D_{4}(P,t)=\rho_0 \int_{0}^{t} F_{Q_1(s), Q_2(t,s)}(t,s) \dot{P}_{s}\ ds,
\end{align}
where $F_{Q_1, Q_2}(t,s)$ is the $3\times 3$ matrix given by
\begin{align*}
F_{Q_1, Q_2}(t,s):=&\langle
\left[
\begin{array}{ccc}
\nabla_{x}W\\
0
\end{array}
\right],\  U(t,s) [H(s)-0]^{-2}    \left[
\begin{array}{ccc}
0\\
 \nabla_{x}W
\end{array}
\right]
\rangle\\
=&\langle
\left[
\begin{array}{ccc}
\nabla_{x}W\\
0
\end{array}
\right],\  e^{(t-s)[H_0+Q_2\cdot \nabla_{x}]} [H_0+Q_1\cdot \nabla_{x}-0]^{-2}    \left[
\begin{array}{ccc}
0\\
 \nabla_{x}W
\end{array}
\right]
\rangle;
\end{align*} the fact that $H_0$ commutes with all components of $\nabla_x$ has been used to find that
$$U(t,s)=e^{(t-s)[H_0+Q_2\cdot \nabla_{x}]},$$
with
\begin{align}\label{eq:defH0}
H_0:=\left[
\begin{array}{ccc}
0& -\Delta\\
-(-\Delta+1) & 0
\end{array}
\right].
\end{align}
In Eq. \eqref{eq:defD4}, the vectors $Q_1,\ Q_2\in \mathbb{R}^3$ are defined as
\begin{align}\label{eq:defP1P2}
Q_1:=P_s\ \text{and}\ Q_2:=\frac{1}{t-s} \int_{s}^{t} P_{s_1}\ ds_1,
\end{align}
with $t>s$.

To prove \eqref{eq:newEstD4}, using \eqref{eq:defD4}, it suffices to show that,
\begin{align}
|F_{Q_1(s),Q_2(t,s)}(s+\tau,s)|\lesssim (1+ \tau)^{-\frac{3}{2}}.\label{eq:Fp1p2}
\end{align}

To express the function $F_{Q_1, Q_2}$ in a convenient form, we diagonalize the matrix operator $H_0$, with the help of the matrix $A$ defined as
\begin{align}\label{eq:difA}
A:=\left[
\begin{array}{ccc}
\sqrt{-\Delta} & \sqrt{-\Delta}\\
i\sqrt{-\Delta+1} & -i\sqrt{-\Delta+1}
\end{array}
\right],
\end{align}
observing that
\begin{align}\label{eq:diag}
A^{-1} H_0 A
=&\left[
\begin{array}{ccc}
iL &0\\
0& -iL
\end{array}
\right],
\end{align}
with $L$ given by
\begin{align}\label{eq:defL}
L:=\sqrt{-\Delta+1}\sqrt{-\Delta}.
\end{align}

After inserting the trivial identity $AA^{-1}=A^{-1}A=Id$ in appropriate places of the expression for
$F_{Q_1, Q_2}$, a direct computation shows that
\begin{align}\label{eq:defFp1p2}
F_{Q_1,Q_2}(t,s)=&C\langle \nabla_{x}W ,\ \frac{\sqrt{-\Delta}}{\sqrt{-\Delta+1}} (L-iQ_1\cdot \nabla_x+i0)^{-2}e^{i(t-s)(L-iQ_2\cdot \nabla_x )}\nabla_{x}W\rangle\nonumber\\
=&C\langle \nabla_{x}V ,\ \frac{(-\Delta)^{2n+\frac{1}{2}}}{\sqrt{-\Delta+1}} (L-iQ_1\cdot \nabla_x+i0)^{-2}e^{i(t-s)(L-iQ_2\cdot \nabla_x )}\nabla_{x}V\rangle
\end{align}
where $C$ is some constant and $W=(-\Delta)^{n}V$, with $V$ and $n$ as in Eq. \eqref{eq:addFriction}.

The difficulties in proving \eqref{eq:Fp1p2} will be discussed in detail in Section \ref{sec:D123}. Some dangerous configurations of momenta $Q_1, Q_2$ will be excluded by showing that the inequality
$|Q_1|\geq |Q_2|$ is very close to being true, and that the vectors $Q_1$ and $\ Q_2$ have essentially the same direction; (see Lemma \ref{LM:est}, below). In excluding dangerous configurations of momenta, we will rely -- implicitly but critically -- on the slow decay of the FGR term and the fast decay of $F_{Q_1, Q_2}$, which can only be established if we require $n$ to be sufficiently large.

To illustrate the ideas underlying our proofs, we limit the present discussion to the following two simple cases:
\begin{align}\label{eq:twocases}
Q_1=Q_2=p(1,0,0),\ p=1,2.
\end{align}

In the following discussion, $W=(-\Delta)^{n}V$, see Eq. (\ref{eq:addFriction}), where $n$ is required
to be sufficiently large.
We propose to sketch a proof of \eqref{eq:Fp1p2}, for $Q_1$ and $\ Q_2$ as in \eqref{eq:twocases}. Hence, in (\ref{eq:defFp1p2}),
 $F_{Q_1,Q_2}(t,s)$ takes the form
\begin{align}
F_{Q_1, Q_2}(t,s)=\tilde{F}_p(\tau)
\end{align}
with $\tau:=t-s$ and $\tilde{F}_p$ defined in the obvious way.
\begin{remark}\label{re:enoughZero}
One message we intend to convey here is that the decay estimate in \eqref{eq:Fp1p2} is sharp, and it can be expected to hold, provided that the exponent $n>0$ is large enough.
\end{remark}

Using the identity
\begin{align*}
[L-ip \partial_{x_1}+i0]^{-2}=-\int_{0}^{\infty}\int_{z}^{\infty}e^{iz_1[L-ip \partial_{x_1}]}\ dz_1 dz,
\end{align*}
we rewrite \eqref{eq:defFp1p2} in the form
\begin{align}\label{eq:Ff}
\tilde{F}_{p}(\tau)=-C\int_{0}^{\infty} dz \int_{z}^{\infty} dz_1 f_p(\tau+z_1),
\end{align}
where $f_p$ is defined as
\begin{align*}
f_p(\eta):
=&\langle \nabla_{x}V ,\ \frac{(-\Delta)^{2n+\frac{1}{2}}}{\sqrt{-\Delta+1}} e^{i\eta[L-ip\ \partial_{x_1} ]}\nabla_{x}V\rangle.
\end{align*}

Fourier transform and a change of variables to polar coordinates then yield
\begin{align}
f_p(\eta)=\int_{0}^{\pi}\int_{0}^{\infty} sin\theta \rho^{4n+5} e^{-i\rho[\sqrt{1+\rho^2}-p cos\theta]\eta} g(\theta) H(\rho)\ d\rho d\theta,
\end{align}
where $H(\rho)$ is a smooth function of rapid decay at $\infty$, and $g(\theta)$ is a polynomial in $sin\theta$ and $cos\theta.$

Integrating by parts in $\theta$, using the identity
\begin{align*}
e^{ip\rho cos\theta \eta} sin\theta= -\frac{1}{ip \rho \eta}\partial_{\theta}e^{ip\rho cos\theta \eta},
\end{align*}
leads us to the expression
\begin{align}
f_p(\eta) =\frac{g(0)}{ip \eta} \int_{0}^{\infty} \rho^{4n+4} e^{-i\rho[\sqrt{1+\rho^2}-p ]\eta} H(\rho)\ d\rho+\cdots
=\frac{g(0)}{ip \eta}\  \tilde{f}_{p}(\eta)+\cdots \label{eq:tif}
\end{align}
where $\tilde{f}_{p}$ is defined in the obvious way, and the dots stand for contributions that decay faster than $\tilde{f}_{p}(\eta)$.

We now analyze the behavior of $f_{p}(\eta)$ for the two choices of $p$ ($p=1,2$) specified in \eqref{eq:twocases}.

For $p=2$, the phase $\rho[\sqrt{1+\rho^2}-p ]$ in the integrand on the right side of \eqref{eq:tif} has only one \textit{non-degenerate} critical point at
 $\rho=\rho_{*}>0$.
A standard stationary phase argument then yields the following asymptotic
behavior of $\tilde{f}_{p}(\eta)$
\begin{align}\label{asy.beh.}
\tilde{f}_{p}(\eta)=C_1 \eta^{-\frac{1}{2}} e^{-i\eta C_2(\rho_{*})}+\cdots,
\end{align}
as $\eta$ tends to $\infty$, where the contribution corresponding to the dots on the right side is subleading, $C_1\in \mathbb{C}$ is a constant depending on $H(\rho_{*})$, and
$$C_2(\rho_{*}):=\rho_{*}\sqrt{1+\rho_{*}^2}-2\rho_{*}\not=0.$$

Eqs. \eqref{eq:tif} and \eqref{asy.beh.} then yield
\begin{align}
f_{p}(\eta)=C_3 e^{-i\eta C_2(\rho_{*})} \eta^{-\frac{3}{2}}+\cdots.
\end{align}
From this equation the desired estimate on $\tilde{F}_{p}(\tau)$, $\tau\geq 0,$ can be inferred by taking into account the oscillatory nature of $e^{-i(\tau+z_1) C_2(\rho_{*})}$ and integrating by parts.

Setting $p=1$, we notice that the only critical point of the phase in the integrand on the right side of \eqref{eq:tif} is at $\rho=0$; it is degenerate, since, for small $\rho$,
$$\rho\sqrt{1+\rho^2}-\rho=\frac{1}{2}\rho^3[1+O(\rho^2)].$$

To obtain an appropriate decay estimate on $\tilde{f}_{p}(\eta)$ we integrate by parts,
using
\begin{align}
e^{-i\eta[\rho\sqrt{1+\rho^2}-\rho]}=\frac{1}{-i\eta }\frac{1}{\sqrt{1+\rho^2}+\frac{\rho^2}{\sqrt{1+\rho^2}}-1}\partial_{\rho}e^{-i\eta [\rho\sqrt{1+\rho^2}-\rho]}.
\end{align}
The singularity of $ \frac{1}{\sqrt{1+\rho^2}+\frac{\rho^2}{\sqrt{1+\rho^2}}-1}=O(\rho^{-2})$ at $\rho=0$
does not cause any problems, thanks to the factor $\rho^{4n+2}$ appearing in the integrand on the right side of \eqref{eq:tif}. If $n$ is chosen large enough we can integrate by parts three times to find
\begin{align}
|\tilde{f}_{p}(\eta)|\lesssim \eta^{-3}.
\end{align}
Inserting this bound into \eqref{eq:tif} and then using \eqref{eq:Ff} we find that
\begin{align}
|\tilde{F}_{p}(\tau)|\lesssim \tau^{-2},
\end{align}
a bound that is better than expected.

To simplify matters, we will choose $n=\frac{3}{4}$ in the remainder of our paper, showing that this value of $n$ is large enough; i.e., we consider a two-body potential $W$ of the form
\begin{align}\label{eq:addFri2}
W=(-\Delta)^{\frac{3}{4}}V,
\end{align}
with $V$ as in \eqref{eq:addFriction} and \eqref{eq:FGR}.

\section{Proof of Proposition \ref{prop:D123} and of Statement (1) in Theorem \ref{THM:main}}\label{sec:D123}
We begin this section with the derivation of an estimate on the term $D_4$ in Eq. \eqref{eq:D3} for
$\dot{P}_t$; (see \eqref{eq:defD4}). This turns out to be the most involved part of our analysis.

In (\ref{eq:defD4}), we write $F_{Q_1,Q_2}$ as
 $$F_{Q_1(s),Q_2(t,s)}(t,s)=C\langle \nabla_{x}V ,\ \frac{(-\Delta)^{2n+\frac{1}{2}}}{\sqrt{-\Delta+1}} (L-iQ_1(s)\cdot \nabla_x+i0)^{-2}e^{i(t-s)(L-iQ_2(t,s)\cdot \nabla_x )}\nabla_{x}V\rangle,$$
see \eqref{eq:defFp1p2}.
It actually turns out to be convenient to study the function $F_{Q_1(s),Q_2(t,s)}(\tau)$ defined by
\begin{align}
F_{Q_1(s),Q_2(t,s)}(\tau):=C\langle \nabla_{x}V ,\ \frac{(-\Delta)^{2n+\frac{1}{2}}}{\sqrt{-\Delta+1}} (L-iQ_1(s)\cdot \nabla_x+i0)^{-2}e^{i\tau(L-iQ_2(t,s)\cdot \nabla_x )}\nabla_{x}V\rangle,
\end{align}
i.e., we treat $\tau\geq 0$ as an independent variable (namely independent of $t,\ s$). We propose to prove that there exists a constant $C$ independent of $t,\ s$ and $\tau$ such that
\begin{align}
|F_{Q_1(s),Q_2(t,s)}(\tau)|\leq C(1+\tau)^{-\frac{3}{2}}.
\end{align}
This obviously implies the desired estimate on $F_{Q_1(s),Q_2(t,s)}(t,s)$ after setting $\tau=t-s.$

Next, we rewrite $F_{Q_1(s),Q_2(t,s)}(\tau)$ in a more convenient form.
Since $W$ is spherically symmetric, there is no loss of generality if we choose the momenta $Q_1$
and $Q_2$, defined in \eqref{eq:defP1P2}, to be given by
\begin{align}\label{eq:formP1P2}
Q_1=(a,b,0),\ Q_2=(\sigma,0,0),
\end{align}
with $\sigma\geq 0$.
By Fourier transformation and after passing to polar coordinates, we find that
\begin{align}
F_{Q_1,Q_2}(\tau)=& \int_{0}^{2\pi}\int_{0}^{\pi}\int_{0}^{\infty}  \frac{\rho^6}{\sqrt{1+\rho^2}} e^{-i\tau Y_{\sigma}(\rho,\theta)}\frac{sin\theta\ |\hat{V}(\rho)|^2}{[G_{a,b}(\rho,\theta,\alpha)-i0]^{2}} g(\theta,\alpha) \  d\rho d\theta d\alpha\label{eq:defY}
\end{align} where $g(\theta,\alpha)$ is a polynomial in $sin\theta$, $cos\theta$ and $e^{\pm i\alpha}$, $W$ is related to $V$ as in \eqref{eq:addFri2},
$$G_{a,b}(\rho,\theta,\alpha):=\sqrt{\rho^2+1}-a cos\theta-b sin\theta cos\alpha,$$
and
$$Y_{\sigma}(\rho,\theta):=\rho\sqrt{\rho^2+1}-\sigma \rho cos\theta .$$

Two types of difficulties arise when the denominator
$G_{a,b}(\rho,\theta,\alpha)$
vanishes at some points, for example when $\theta=0$, $a>1$ and $\rho=\sqrt{a-1}$: (1) A minor one is encountered if these zeros of $G_{a,b}(\rho,\theta,\alpha)$ are \textit{not} critical points of the phase
$\rho\sqrt{\rho^2+1}-\sigma \rho cos\theta$. In this case, the difficulty can be resolved as in \cite{MR1681113}.
(2) A more serious difficulty is met when the denominator vanishes (or almost vanishes) at points
$(\rho,\theta,\alpha)=(\rho_{*},\theta_{*},\alpha_{*})$, where $(\rho_{*},\ \theta_{*})$ are critical points of the phase,  $Y_{\sigma}(\rho,\theta)$, in the integrand on the right side of \eqref{eq:defY}
. Then the decay of $F_{Q_1,Q_2}(\tau)$ in $\tau$ may be slower than desirable. As an example, we notice that
the function given by $\int_{-\infty}^{\infty} e^{ik^2 \tau} |k|^{-\frac{1}{2}}\ dk$ decays significantly more slowly than the function given by $\int_{-\infty}^{\infty} e^{ik^2 \tau} \ dk$, and this is due to the singularity of $|k|^{-1/2}$ at $k=0$, which is a critical point of the phase $k^2.$

It turns out that, for $\sigma>1$, the phase $Y_{\sigma}(\rho,\theta)$ has two critical points:
\begin{align}\label{eq:critical}
(\rho,\theta)=(\zeta,0),\ (0,\eta)
\end{align}
where $\zeta>0$ and $\eta\in (0,\pi)$ are solutions to the equations
\begin{align}\label{eq:zeta}
\sqrt{1+\zeta^2}+
\frac{\zeta^2}{\sqrt{1+\zeta^2}}=\sigma \text{    }\text{and}\ \ cos(\eta)=\frac{1}{\sigma}.
\end{align}

At the critical point $(0,\eta)$, the denominator vanishes. For example, if $Q_1=Q_2$ then, by \eqref{eq:zeta},
\begin{align*}
G_{a,b}(\rho,\theta,\alpha)|_{\rho=0,\ \theta=\eta}=1-\sigma cos\eta=0.
\end{align*}
However, the factor $\rho^6$ in the integrand on the right side of \eqref{eq:defY} offsets the
singularity of the factor $G_{a,b}^{-2}$ at $\rho=0$.

The critical point $(\zeta,0)$ of the phase in the integrand on the right side of \eqref{eq:defY} does
not do any harm to the decay of  $F_{Q_1,Q_2}(\tau)$ either, thanks to the fact that the following two statements are ``very close to being correct'':
\begin{align}\label{eq:alDePa}
|Q_1|\geq |Q_2|>1,\ \text{and}\ Q_1\ \text{is parallel to } Q_2.
\end{align} To see that $|G_{a,b}^{-1}(\rho,\theta,\alpha)|_{\rho=\zeta,\theta=0}$ is appropriately bounded in these cases, we suppose that $|Q_1|\geq |Q_2|$ and $Q_1$ is parallel to $Q_2$, which by \eqref{eq:formP1P2} implies $a\geq \sigma>1$ and $b=0.$
Then we have the following upper bound
\begin{align*}
|G_{a,b}^{-1}(\rho,\theta,\alpha)|_{\rho=\zeta,\theta=0}
=&\frac{1}{|a-\sqrt{1+\zeta^2}|}\\
=&\frac{1}{|a-\sigma+\sigma-\sqrt{1+\zeta^2}|}\\
=&\frac{1}{|a-\sigma+\frac{\zeta^2}{\sqrt{1+\zeta^2}}|}\\
\leq & \frac{\sqrt{1+\zeta^2}}{\zeta^2},
\end{align*} where in the second but last step we have used the first equation in \eqref{eq:zeta}.

Recall that $Q_1$ and $Q_2$ are related to $P_t$ by \eqref{eq:defP1P2}. Conditions (I) and (II) of Lemma \ref{LM:est}, below, will turn out to suffice to prove the desired bound on $F_{Q_1,Q_2}$. More detailed information will be provided in Lemma \ref{LM:q1q2}.

\begin{lemma}\label{LM:est}
Suppose the following two conditions hold on some time interval $[0,\ T]$.
\begin{itemize}
\item[(I)]  $\frac{d}{dt} |P_t| \leq 0$, for any time $t$ with $\rho_0^{-\frac{1}{10}}\leq t \leq T$.
\item[(II)] The momentum $P_{.}$ has the properties
\begin{align*}
|P_t|-1\geq \rho_0^{\frac{1}{4}} (1+\rho_0 t)^{-\frac{2}{5}},\text{   }\text {and}
\ |\frac{P_t}{|P_t|}-\frac{P_s}{|P_s|}|\leq \rho_0^{\frac{3}{4}} (1+\rho_0 s)^{-\frac{2}{5}},
\end{align*}
for arbitrary times $s$ and $t$, with $0\leq s\leq t \leq T$.
\end{itemize}
Then the function $F_{Q_1,Q_2}(\tau)$ satisfies the decay estimate
\begin{align}\label{eq:estKts}
|F_{Q_1, Q_2}(\tau)|\lesssim (1+\tau)^{-\frac{3}{2}},
\end{align}
for any $\tau\geq 0$.
\end{lemma}
Conditions (I) and (II) required in Lemma \ref{LM:est} are the subject of Lemma \ref{LM:bootstrap}, below, which will be proven in Subsection \ref{subsec:initial}. The decay estimate \eqref{eq:estKts} is proven in Sections \ref{sec:GeqM} and \ref{sec:A}, where two different regimes will have to be considered separately: $|Q_2|\geq 1+10 \tau^{-\frac{2}{3}}$ and $|Q_2|\leq 1+10 \tau^{-\frac{2}{3}}$. Some technicalities will be proven in various appendices.

Next, we return to estimating the function $D_4$ given in \eqref{eq:defD4}.
Our bound on $F_{Q_1, Q_2}(\tau)$ implies that
\begin{align}
|D_4|\lesssim &\rho_0 \int_{0}^{t}(1+t-s)^{-\frac{3}{2}} |\dot{P}_s|\ ds\nonumber\\
\lesssim & \Omega(t) \rho_0^2 \int_{0}^{t}(1+t-s)^{-\frac{3}{2}} (1+\rho_0 s)^{-\frac{7}{5}}\ ds\nonumber\\
\lesssim & \Omega(t)\rho_0^{\frac{2}{5}} \tilde{D},\label{eq:ttD}
\end{align}
where $\tilde{D}$ is defined by
 $$\tilde{D}:=\int_{0}^{t} (1+t-s)^{-\frac{3}{2}} (\rho^{-1}_0+s)^{-\frac{7}{5}}\ ds,$$
 and the function $\Omega$ is defined as
\begin{align}\label{eq:majorant}
\Omega(t):=\rho_{0}^{-1}\max_{0\leq s\leq t}(1+\rho_0 s)^{\frac{7}{5}} |\dot{P}_s|.
\end{align}

Using that $ (\rho^{-1}_0+s)^{-\frac{7}{5}}\leq \rho_{0}^{\frac{7}{5}}$ one obtains the bound
\begin{align*}
\tilde{D}\leq & \int_{0}^{t} (1+t-s)^{-\frac{3}{2}}\rho_0^{\frac{7}{5}}\ ds
\lesssim   \rho_0^{\frac{7}{5}}.
\end{align*}
Next, using that $ (\rho^{-1}_0+s)^{-\frac{7}{5}}\leq (1+s)^{-\frac{7}{5}}$ and considering separately the two domains $s\geq \frac{t}{2}$ and $s\leq \frac{t}{2}$, one observes that
\begin{align*}
\tilde{D}\leq &\int_{0}^{t} (1+t-s)^{-\frac{3}{2}} (1+s)^{-\frac{7}{5}}\ ds
\lesssim   (1+t)^{-\frac{7}{5}}.
\end{align*}
Taking the minimum of these two bounds, we conclude that
 $$\tilde{D}\leq \min\{\rho_0^{\frac{7}{5}},\ (1+t)^{-\frac{7}{5}} \}\lesssim (\rho_0^{-1}+t)^{-\frac{7}{5}}.$$

Plugging this bound into \eqref{eq:ttD}, we obtain the desired estimate:
\begin{lemma}\label{lemmaD_4}
Suppose that conditions (I) and (II) of Lemma \ref{LM:est} hold. Then
\begin{align}\label{eq:estD4}
|D_4(t)|\lesssim  \rho_0^2 \Omega(t) (1+\rho_0 t)^{-\frac{7}{5}}.
\end{align}
\end{lemma}

Next, we analyze the terms $D_2$ and $D_3$ in the equation of motion \eqref{eq:D3} for $P_t$.
\begin{lemma}\label{LM:d2d3}
Suppose that conditions (I) and (II) in Lemma \ref{LM:est} hold. Then
\begin{align}
|D_2(t)|\lesssim  &(1+t)^{-\frac{3}{2}}\sqrt{\rho_0} \|(1+x^2)^{2} {\bf{h}}_0\|_{\mathcal{H}^2}\lesssim \rho_0(1+t)^{-\frac{3}{2}},\\
|D_3(t)|\lesssim    &\rho_0 (1+t)^{-\frac{3}{2}}.
\end{align}
(In the first inequality, we use the smallness of the initial condition, viz.
${\bf{h}}_0 =O(\rho_{0}^{\frac{1}{2}}).)$
\end{lemma}
The proof of this lemma is considerably easier than that of Lemma \ref{lemmaD_4}, and we omit it.

Next, we turn to the proof of Statement (1) of our Main Result, Theorem \ref{THM:main}.

\subsection{Proof of Statement (1) of the Theorem \ref{THM:main}}
The goal of this section is to prove statement (1) in Theorem \ref{THM:main}. Our proof is based on Lemmas  \ref{LM:est}-\ref{LM:d2d3} and Lemma \ref{LM:bootstrap}, below, and involves a bootstrap argument.

We recall that, on the time interval $[0,\rho_{0}^{-\frac{1}{10}}]$, the solution $P_t$ has already been studied in Proposition \ref{prop:smallVar}. In this section we focus our attention on the behavior of
$P_t$, for $t \in [\rho_0^{-\frac{1}{10}},\infty)$. We first analyze the behavior of $P_t$ and show that
$|P_t|>1$, for
$t \in \lbrack \rho_{0}^{-\frac{1}{10}},T\rbrack$ with $T<\infty$, and then employ a bootstrap argument to show that
$T$ can actually be let tend to $\infty$.
Recall that the function $\Omega(t)$ has been defined in \eqref{eq:majorant}.
\begin{lemma}\label{LM:bootstrap}
There exists a time $T$ satisfying $T>\rho_{0}^{-\frac{1}{10}}$ such that
Conditions (I) and (II) in Lemma \ref{LM:est}, as well as the two inequalities
\begin{align}\label{eq:assumBo}
\Omega(t)\leq \rho_0^{-\frac{1}{2}}\ \text{and}\ [|P_t|-1]^{\frac{7}{2}}\geq \rho_0^{\frac{1}{20}} (1+\rho_0 t)^{-\frac{7}{5}},
\end{align}
hold for all $t \in [0,T]$.
\end{lemma}
The proof of this lemma can be found at the end of this section.

Next, we use the results in Lemmas \ref{LM:est}-\ref{LM:bootstrap} to solve the equation of motion \eqref{eq:D3} for $P_t$, for
$t \in [0,\ T]$.

The validity of conditions (I) and (II) enables us to apply the results in Lemmas
\ref{LM:est}\text{   }--\text{   }\ref{LM:d2d3}, which along with \eqref{eq:assumBo} imply that
\begin{align}\label{eq:remainder}
| D_2(P,t)+D_3(P,t)+D_4(P,t)|\lesssim \rho_0 (1+t)^{-\frac{3}{2}} +\rho_0^{\frac{3}{2}} (1+\rho_0 t)^{-\frac{7}{5}}.
\end{align}

To bound the term $D_1(P_t)$, we use Lemma \ref{LM:est1} and the second inequality in \eqref{eq:assumBo} and find that
\begin{align}\label{eq:ideal2}
|D_1(P_t)|\geq 2| D_2(P,t)+D_3(P,t)+D_4(P,t)|,
\end{align}
for $t \in [\rho_{0}^{-\frac{1}{10}},T]$.
This enables us to solve the equation for $\frac{d}{dt}|P_t|$, see \eqref{eq:lower}-\eqref{eq:lower3}, for
$t \in [\rho_{0}^{-\frac{1}{10}},T]$, and to prove the crucial lower and upper bounds on
 $|P_t|-1$:
\begin{align}
C_1(1+\rho_0 t)^{-\frac{2}{5}}\leq |P_t|-1 \leq C_2(1+\rho_0 t)^{-\frac{2}{5}}, \label{eq:below}
\end{align}
for some positive constants $C_1$ and $C_2$. Plugging these bounds into the equation for
$\frac{d}{dt}|P_t|$ in ~\eqref{eq:lower}, we obtain that
\begin{align}\label{eq:tPtLU}
-C_3 \rho_0(1+\rho_0 t)^{-\frac{7}{5}}\leq \frac{d}{dt}|P_t|\leq -C_4 \rho_0(1+\rho_0 t)^{-\frac{7}{5}},
\end{align}
for some constants $C_3,\ C_4>0$.
After controlling the magnitude of $P_t$, we study its direction, $\frac{P_t}{|P_t|}.$
From the equation of motion \eqref{eq:D3} for $P_t$ we derive an equation for
$\frac{d}{dt}(\frac{P_t}{\vert P_t \vert})$:
\begin{align}\label{Eq.dir.}
\frac{d}{dt}(\frac{P_t}{|P_t|}) =\frac{1}{|P_t|} \sum_{k=2}^{4}D_k-[P_t\cdot \displaystyle\sum_{k=2}^{4}D_k]\frac{1}{|P_t|^{\frac{3}{2}}} P_t.
\end{align}
Note that a term proportional to $D_1$ does not appear on the right side of this equation, because $D_1(P_t)$ is parallel to $P_{t}$; (see \eqref{eq:D1}).
Applying our estimates on $D_k, k=2,3,4,$ in \eqref{eq:remainder} on the right side of Eq. \eqref{Eq.dir.} and using \eqref{eq:tPtLU},
we obtain that
\begin{align*}
|\frac{d}{dt}(\frac{P_t}{|P_t|})|\lesssim \rho_0(1+t)^{-\frac{3}{2}}+\rho_0^2 (1+\rho_0 t)^{-\frac{7}{5}}.
\end{align*}
Integrating $\frac{d}{dt}(\frac{P_t}{|P_t|})$ from $s$ to $t$, with $0\leq s\leq t,$ and using this bound, we find that
\begin{align}\label{eq:convDi}
|\frac{P_s}{|P_s|}-\frac{P_t}{|P_t|}|\lesssim \rho_0 (1+\rho_0 s)^{-\frac{2}{5}}.
\end{align}

Next, we show that the desired estimates \eqref{eq:below}, \eqref{eq:convDi} hold for all $t \in [0,\infty)$ by proving that the maximal value of $T$ for which Lemma \ref{LM:bootstrap} holds is $T = \infty$, and then repeating the arguments above:
Suppose the maximal value of $T$ is given by some $T_{*}<\infty.$ Then we may apply \eqref{eq:below}, \eqref{eq:tPtLU} and \eqref{eq:convDi} and use arguments similar to those used in the proof of Lemma \ref{LM:bootstrap}, below, to extend the validity of Lemma \ref{LM:bootstrap} to some larger time $T_{**}>T_{*}$. Consequently the maximal value of $T$ is $\infty.$

Before turning to our proof of Lemma  \ref{LM:bootstrap} we complete the proof of Statement (1) of Theorem \ref{THM:main}: the convergence of $|P_t|$ is implied by \eqref{eq:below}, the convergence of direction of $P_t$ is a consequence of \eqref{eq:convDi}, and Eq. \eqref{Pdot} for $\dot{P_t}$ follows from our lower bound on $-D_1(P_t)$ and the upper bounds on
$|D_2+D_3+D_4|.$
\subsubsection{Proof of Lemma \ref{LM:bootstrap}}\label{subsec:initial}
We start by estimating $P_t$, $\dot{P}_t$ and $\frac{d}{dt}|P_t|$, for $t$ in the interval
$[0,\rho_0^{-\frac{1}{10}}]$, by applying the results in Proposition \ref{prop:smallVar}. Afterwards, we extend them to a larger time-interval.

Proposition \ref{prop:smallVar} implies that, for $t \in [0, \rho_0^{-\frac{1}{10}}]$,
\begin{align}\label{eq:omegat}
\Omega(t)\lesssim 1.
\end{align}
It also shows that, for arbitrary times $t$ and $s$ satisfying
$\rho_0^{-\frac{1}{10}}\geq t\geq s\geq 0$,
\begin{align}
|P_t|-1\gtrsim (1+\rho_0 t)^{-\frac{2}{5}},\ |\frac{P_t}{|P_t|}-\frac{P_s}{|P_s|}|\lesssim \rho_0^{\frac{9}{10}} (1+\rho_0 s)^{-\frac{2}{5}}.
\end{align}

Next, we turn to verifying condition (I) in Lemma \ref{LM:est}, for $t$ of order $O(\rho_0^{-\frac{1}{10}})$. Using the bound \eqref{eq:estD4} and Lemma \ref{LM:est1}, \eqref{eq:D1}, we obtain that
\begin{align}
|\sum_{k=2}^{4}D_k(t_0)|\lesssim  \rho_0^{\frac{11}{10}},\
|D_1(t_0)|\gtrsim   \rho_0,
\end{align}
for $t_0=\rho_{0}^{-\frac{1}{10}}$.
When inserted on the right side of inequality \eqref{eq:lower2} for $\frac{d}{dt}|P_t|$ one finds that
\begin{align}\label{eq:Ptupper}
\lbrack-\frac{d}{dt}|P_t|\rbrack_{t=t_{0}=\rho_{0}^{-\frac{1}{10}}}\gtrsim \rho_0,
\end{align}

The results in \eqref{eq:omegat}-\eqref{eq:Ptupper}, for the time interval $[0,\rho_{0}^{-\frac{1}{10}}]$, are stronger than those in Lemma \ref{LM:bootstrap}. Hence, by continuity, a weaker version holds in a somewhat larger time interval; i.e., there exists a time interval $[0,T]$, with $T>\rho_0^{-\frac{1}{10}}$, on which Lemma \ref{LM:bootstrap} holds.

This completes the proof of Lemma \ref{LM:bootstrap}.

\section{The State of the Bose Gas, as $t\rightarrow \infty$ -- Proof of Statement (2) of Theorem \ref{THM:main}}\label{sec:formState}

To prove the convergence of the solution of Eq. \eqref{eq:field} for the condensate wave function,
$\beta_t$, of the Bose gas, it is not convenient to use the decomposition in \eqref{eq:decom1}, because the presence of a singularity in $[H(t)-0]^{-1}=[H_0+P_t\cdot \nabla_x-0]^{-1}$, for $|P_t|>1,$ would make it cumbersome to find an appropriate function space (for $\delta_t$)
to work with.
Instead, we propose to find an equation for $\beta_t$ that takes into account the fact -- proven in Sect. 3.1 -- that
$P_t\rightarrow P_\infty$, as $t\rightarrow\infty$, with $|P_{\infty}|= 1$.

We define a vector function $\xi_{t}:\ \mathbb{R}^3\rightarrow \mathbb{R}^2$ by
\begin{align}\label{eq:vForm2}
\xi_{t}  (x-P_{\infty}t)=\left[
\begin{array}{ccc}
\text{Re}\beta_{t}(x)\\
\text{Im}\beta_{t}(x)
\end{array}
\right].
\end{align}
Then Eq. \eqref{eq:field} reads
\begin{align}
\dot{\xi}_{t}=&H_{\infty} \xi_{t}-\sqrt{\rho_0} \left[
\begin{array}{ccc}
0\\
W^{Y_t}
\end{array}
\right], \label{eq:field3}
\end{align}
where $H_{\infty}$ is the $2\times 2$ matrix operator given by
\begin{align}
H_{\infty}:=\left[
\begin{array}{ccc}
P_{\infty}\cdot \nabla_{x}& -\Delta\\
-(-\Delta+1) & P_{\infty}\cdot\nabla_{x}
\end{array}
\right],
\end{align}
and
\begin{align}\label{eq:defYt}
Y_t:=X_t-P_{\infty}t.
\end{align}
We decompose the function $\xi_t$ into two parts:
\begin{align}\label{eq:decom3}
\xi_{t}=\sqrt{\rho_0} H_{\infty}^{-1} \left[
\begin{array}{ccc}
0\\
W^{Y_t}
\end{array}
\right]+\eta_{t}
\end{align}
We observe that $H_{\infty}^{-1}$ is well defined due to the fact that $|P_{\infty}|= 1.$
Eq. \eqref{eq:field3} and the decomposition in \eqref{eq:decom3} yield an evolution equation for $\eta_t$:
\begin{align}
\dot{\eta}_t=H_{\infty}\eta_t+\sqrt{\rho_0} H_{\infty}^{-1} [P_t-P_\infty]\cdot \nabla_{x}\left[
\begin{array}{ccc}
0\\
W^{Y_{t}}
\end{array}
\right].
\end{align}
Applying Durhamel's Principle, we find that
\begin{align}
\eta_t=&e^{H_{\infty}t}\eta_0+\sqrt{\rho_0} \int_{0}^{t} e^{(t-s)H_{\infty}}H_{\infty}^{-1} [P_s-P_\infty]\cdot \nabla_{x}\left[
\begin{array}{ccc}
0\\
W^{Y_{s}}
\end{array}
\right]\ ds\nonumber\\
=&e^{H_{\infty}t}\xi_0-\sqrt{\rho_0}e^{H_{\infty}t} H_{\infty}^{-1}\left[
\begin{array}{ccc}
0\\
W^{Y_{0}}
\end{array}
\right]\nonumber\\
&+\sqrt{\rho_0} \int_{0}^{t} e^{(t-s)H_{\infty}}H_{\infty}^{-1} [P_s-P_\infty]\cdot \nabla_{x}\left[
\begin{array}{ccc}
0\\
W^{Y_{s}}
\end{array}
\right]\ ds\label{eq:threeTerms}
\end{align}
The different terms on the right side of Eq. \eqref{eq:threeTerms} are estimated in the following lemma.

\begin{lemma}\label{LM:HInfty}
For any $t\geq 0,$
\begin{align}
\|e^{H_{\infty}t}\eta_0\|_{\infty}\lesssim & (1+t)^{-1}\|(1+|x|)^4\eta_{0}\|_{\mathcal{H}^{2}},\\
\|e^{H_{\infty}t} H_{\infty}^{-1} \left[
\begin{array}{ccc}
0\\
W^{Y_{0}}
\end{array}
\right]\|_{\infty}\lesssim & (1+t)^{-1},\label{eq:singular}\\
\|e^{H_{\infty}t} H_{\infty}^{-1} \nabla_{x}\left[
\begin{array}{ccc}
0\\
W^{Y_{s}}
\end{array}
\right]\|_{\infty}\lesssim &(1+t)^{-1}.
\end{align}
\end{lemma}
The proof of Lemma \ref{LM:HInfty} can be found in Section \ref{sec:Hinf}.

Applying this lemma we obtain that
\begin{align}\label{eq:pspinfty}
\|\eta_t\|_{\infty}\lesssim (1+t)^{-1}[\sqrt{\rho_0}+\|(1-\Delta)^{2}\xi_0\|_{L^{1}}]+\sqrt{\rho_0}\int_{0}^{t}(1+t-s)^{-1}|P_s-P_{\infty}|\ ds.
\end{align}
Since $|\dot{P}_t|\lesssim \rho_0(1+\rho_0 t)^{-\frac{7}{5}}$, see \eqref{eq:tPtLU}, we have control over $|P_t-P_{\infty}|$, namely
\begin{align*}
|P_t-P_{\infty}|\lesssim (1+\rho_0 t)^{-\frac{2}{5}}.
\end{align*}
 This bound, together with \eqref{eq:pspinfty}, obviously implies that
\begin{align}\label{eq:etaDecay}
\|\eta_t\|_{\infty}\rightarrow 0,\ \text{as} \ t\rightarrow \infty.
\end{align}

Next we show that this implies the desired result \eqref{eq:betaconv} by relating $\eta_t$ to $\beta_t-\beta_{\infty}(\cdot-X_t)$. Using the definition of $\xi$ in \eqref{eq:vForm2} and its decomposition in \eqref{eq:decom3}, and recalling the definition of $Y_t$ in \eqref{eq:defYt}, we find that
\begin{align}
\left[
\begin{array}{ccc}
\text{Re}\beta_{t}(x)\\
\text{Im}\beta_{t}(x)
\end{array}
\right]=\sqrt{\rho_0} H_{\infty}^{-1} \left[
\begin{array}{ccc}
0\\
W^{Y_t+P_{\infty}t}
\end{array}
\right]+\eta_{t}(x-P_{\infty}t)=\sqrt{\rho_0} H_{\infty}^{-1} \left[
\begin{array}{ccc}
0\\
W^{X_t}
\end{array}
\right]+\eta_{t}(x-P_{\infty}t).\label{eq:original}
\end{align}
Defining $\beta_{\infty}:=\sqrt{\rho_0} H_{\infty}^{-1} \left[
\begin{array}{ccc}
0\\
W
\end{array}
\right]$, we observe that this identity and property \eqref{eq:etaDecay} complete the proof of our main result, Theorem \ref{THM:main}; (with the proof of Lemma \ref{LM:est1} postponed to Section \ref{sec:LemEst}, the one of Lemma \ref{LM:est} postponed
to Sections \ref{sec:GeqM} and \ref{sec:A} and the one of Lemma \ref{LM:HInfty} to
Section \ref{sec:Hinf}).

\section{Proof of Equation \eqref{eq:D1}, Lemma \ref{LM:est1}}\label{sec:LemEst}
By definition, the term $D_1$ appearing in the equation of motion \eqref{eq:D3}
for the particle is given by
$$D_1(P)= \rho_{0}\langle\left[
\begin{array}{ccc}
\nabla_{x}W\\
0
\end{array}
\right], [H_0+P\cdot \nabla_{x}-0]^{-1} \left[
\begin{array}{ccc}
0\\
W
\end{array}
\right]\rangle$$
with $H_0:=\left[
\begin{array}{ccc}
0& -\Delta\\
-(-\Delta+1) & 0
\end{array}
\right], $
(see \eqref{eq:defH0}).

We first notice that if $|P|\leq 1$ then $$D_{1}(P)\equiv 0.$$
To see this one uses the fact that, for $|P| \leq 1$, the operator $H_0+P\cdot \nabla_{x}$ is invertible and then one determines the form of its inverse.
It follows that the function $\Lambda(|P|)$ in Eq. \eqref{eq:D1} vanishes identically, for $|P|\leq 1$.

In the remainder of this section, we assume that $|P|>1$.

A simple symmetry argument shows that the vector $D_1(P)$ is parallel to $P$, for all
$P\in \mathbb{R}^{3}$. To see this we choose two arbitrary vectors $Q_1, \ Q_2\in\mathbb{R}^3$, with $Q_1\perp Q_2$, and show that
\begin{align*}
\rho_{0}^{-1}Q_1\cdot D_1(Q_2)
=& Q_1\ \cdot\ \langle\left[
\begin{array}{ccc}
\nabla_{x}W\\
0
\end{array}
\right], [H_0+Q_2\cdot \nabla_{x}-0]^{-1} \left[
\begin{array}{ccc}
0\\
W
\end{array}
\right]\rangle\\
=&\langle\left[
\begin{array}{ccc}
Q_1\cdot\nabla_{x}W\\
0
\end{array}
\right], [H_0+Q_2\cdot \nabla_{x}-0]^{-1} \left[
\begin{array}{ccc}
0\\
W
\end{array}
\right]\rangle=0,
\end{align*}
(We recall that $W$ and $H_0$ are invariant under rotations of
$\mathbb{R}^{3}$. Without loss of generality one may therefore assume that $Q_1=|q_1|(1,0,0)$ and $Q_2=|q_2|(0,1,0).$ The above expression is then seen to vanish, because it is given by an integral over a function that is odd in the $1-$ direction.)

It follows that $D_{1}(P)$ is of the form
\begin{align}
D_{1}(P)=\rho_0 \frac{P}{|P|} \tilde{D}_1(P),
\end{align}
where $\tilde{D}_1$ is a scalar function given by
\begin{align}
\tilde{D}_1(P):=&\langle\left[
\begin{array}{ccc}
\frac{P}{|P|}\cdot\nabla_{x}W\\
0
\end{array}
\right], [H_0+P\cdot \nabla_{x}-0]^{-1} \left[
\begin{array}{ccc}
0\\
W
\end{array}
\right]\rangle\nonumber\\
=&\langle\left[
\begin{array}{ccc}
\partial_{x_3}W\\
0
\end{array}
\right], [H_0+|P|\ \partial_{x_3}-0]^{-1} \left[
\begin{array}{ccc}
0\\
W
\end{array}
\right]\rangle.
\end{align}

We propose to derive an explicit expression for $\tilde{D}_1$.
To render our calculation more transparent we diagonalize $H_0$, as in \eqref{eq:difA}-\eqref{eq:defL},
and find that
\begin{align}
\tilde{D}_1(P)=-2Re \langle \partial_{x_3}W,\ \frac{\sqrt{-\Delta}}{\sqrt{-\Delta+1}}[L-i |P|\partial_{x_3}+i0]^{-1}W\rangle,
\end{align}
where $L$ is as in \eqref{eq:defL}.
Fermi-Golden-Rule terms similar to $\tilde{D}_1$ have come up in many different contexts and have been used for purposes similar to ours in ~\cite{MR1681113,TsaiYau02,BuSu,GaSi2007, MR2182081}:
 By Fourier transformation and then passing to polar coordinates, one sees that
\begin{align*}
\tilde{D}_1(P)=&-2\pi Re\ i \int_{0}^{\pi} \int_{0}^{\infty} \frac{\rho^{4+4n} |\hat{V}(\rho)|^2\ cos\theta\ sin\theta}{\sqrt{1+\rho^2}}[\rho \sqrt{1+\rho^2}-|P|\rho\ cos\theta+i0]^{-1}\ d\rho d\theta\\
=&2\pi Im\  \int_{0}^{\frac{\pi}{2}} \int_{0}^{\infty} \frac{\rho^{3+4n} |\hat{V}(\rho)|^2\ sin\theta}{\sqrt{1+\rho^2}}[ \frac{\sqrt{1+\rho^2}}{cos\theta}-|P|+i0]^{-1}\ d\rho d\theta,
\end{align*}
where, in the expression on the right side of this equation, we have changed the integration region of
the variable $\theta$ to $[0,\frac{\pi}{2}]$. This is justified by observing that the contribution corresponding to the integration domain $[\frac{\pi}{2},\ \pi]$ vanishes, because $cos\theta<0$. (We have used that $W=(-\Delta)^{n}V$, see \eqref{eq:addFri2}.)

We now prove \eqref{eq:D1}.
If $|P|\geq 1+\epsilon_0$, with $\epsilon_0>0$, it is easy to apply Lemma \ref{LM:FGR}, below, and use the Fermi-Golden-Rule condition \eqref{eq:FGR} to prove that there is a $\delta_0(\epsilon_0)>0$
such that
$$\tilde{D}_1\leq -\delta_0(\epsilon_0).$$

Next, we determine the behavior of $\tilde{D}_1(P)$ when $q:=|P|-1>0$ is very close to 0,
($q\searrow 0$).
The integration region contributing to $\tilde{D}_1$ is the region where the function
 $$\frac{\sqrt{1+\rho^2}}{cos\theta}-1=\frac{1}{2}(\rho^2+\theta^2)+O(\rho^4+\theta^4)$$
 is small, i.e., where $\rho$ and $\theta$ are small.
Hence
\begin{align}
\tilde{D}_1= &2\pi Im \int_{0}^{\frac{\pi}{2}} \int_{0}^{\infty}  \frac{\rho^{3+4n} |\hat{V}(\rho)|^2\ sin\theta}{\sqrt{1+\rho^2}}[ \frac{\sqrt{1+\rho^2}-cos \theta}{cos\theta}-q+i0]^{-1}\times\\
                        &\times \chi( \frac{2[\sqrt{1+\rho^2}-cos \theta]}{cos\theta})\ d\rho d\theta,
\end{align}
where $\chi$ is a smooth cutoff function satisfying $\chi(s)=1$, for $s\leq \frac{1}{4}$, and $\chi(s)=0$ if $s\geq \frac{1}{2}.$ To simplify this expression, we introduce new variables, $r$ and $\alpha$, by setting
\begin{align}
r^2:=2\frac{\sqrt{1+\rho^2}-1}{cos\theta}=\rho^2[1+O(\rho^2+\theta^2)],\ \ \alpha^2:=2\frac{1-cos\theta}{cos\theta}=\theta^2(1+O(\theta)^2).
\end{align}
We then find that
\begin{align}
\tilde{D}_1=4\pi Im \int_{0}^{\infty}\int_{0}^{\infty} r^{3+4n}\ \alpha\ H(\alpha, r) [\alpha^2+r^2-2q+i0]^{-1}\ dr\ d\alpha,
\end{align}
where $H(\alpha, r)$ is a smooth real-valued function of rapid decay, with $H(0,0)=|\hat{V}(0)|^2=1.$
By re-scaling variables, $r\rightarrow \sqrt{2q}r$ and $\alpha\rightarrow \sqrt{2q} \alpha$, passing to polar coordinates and applying Lemma \ref{LM:FGR}, below, Eq. \eqref{eq:D1} is seen to follow.

\subsection{A simple identity}
\begin{lemma}\label{LM:FGR}
Suppose $f:\ [0,\infty)\rightarrow \mathbb{R}$ is a real-valued, bounded continuous function. Then
\begin{align}
Im \int_{0}^{\infty}[\rho-1+i0]^{-1} f(\rho)\ d\rho=-\pi f(1).
\end{align}
\end{lemma}
\begin{proof}
Setting $\rho-1=:r$, one verifies that
\begin{align*}
Im \int_{0}^{\infty}[\rho-1+i0]^{-1} f(\rho)\ d\rho=& Im \int_{-1}^{\infty} [r+i0]^{-1}f(1+r) \ dr\\
=& \frac{1}{2i} \displaystyle\lim_{\epsilon\rightarrow 0^+} \int_{-1}^{\infty}\bigg( [r+i\epsilon]^{-1}- [r-i\epsilon]^{-1}\bigg)f(1+r)\ dr\\
=& - \displaystyle\lim_{\epsilon\rightarrow 0^+} \int_{-1}^{\infty}\frac{\epsilon}{r^2+\epsilon^2}f(1+r)\ dr\\
=&- \displaystyle\lim_{\epsilon\rightarrow 0^+} \int_{-\frac{1}{\epsilon}}^{\infty}\frac{1}{1+r^2} f(1+\epsilon r)\ dr\\
=& -\pi f(1),
\end{align*}
where in the second but last step we have re-scaled the integration variable, $r\rightarrow \epsilon r $, and in the last step we have  used that $$\int_{-\infty}^{\infty} \frac{1}{1+r^2}\ dr=\pi. $$
\end{proof}

In the remaining sections we will have to derive various decay estimates that have been assumed so far.

%%%%%%%%%%%%%%%%%%%%%%%%%%%%%%%%%%%%%%%%%%%%%%
%%%%%%%%%%%%%%%%%%%%%%%%%%%%%%%%%%%%%%%%%%%%%%%

\section{Proof of Lemma \ref{LM:est} when $|Q_2|> 1+10 \tau^{-\frac{2}{3}}$}\label{sec:GeqM}
In this section we analyze the decay of the function $F_{Q_1,Q_2}(\tau)$ in time, $\tau$, which will then yield Lemma \ref{LM:est}.
We may assume that $\tau>0$ is large. (For $\tau\sim O(1)$, one shows that
$\vert F_{Q_1,Q_2}(\tau)\vert$ is bounded, and this follows from a change of the contour of integration
introduced in the next section. We omit details.)

We start with an analysis of the factor $[G_{a,b}(\rho,\theta,\alpha)-i0]^{-2}$ on the right side of expression \eqref{eq:defY} for $F_{Q_1,Q_2}(\tau)$ in a neighborhood of the critical point $(\rho,\theta)=(\zeta,0)$ of the phase $Y_{\sigma}$. (Recall the definition of $\zeta$ in \eqref{eq:zeta}, and recall our discussion of the importance of controlling $G^{-1}$ near critical points of $Y_{\sigma}$ at the beginning of Section \ref{sec:D123}.)

Before we can state our results we must introduce two constants, $\zeta_0$ and
$R$: The constant $\zeta_0>0$ is the solution of the equation
\begin{align}\label{eq:zeta0}
\sqrt{1+\zeta_0^2}+\frac{\zeta_0^2}{\sqrt{1+\zeta_0^2}}=\sqrt{a^2+b^2}=|Q_2|>1,
\end{align}
where $a$ and $b$ have been introduced in Eq.\eqref{eq:formP1P2},  and $R$ is defined by
\begin{align}\label{eq:defR}
R:=\frac{1}{5} \frac{\zeta_0}{(1+\zeta_0^2)^{\frac{1}{2}}}.
\end{align}
(Recall that $\vert Q_2\vert=\sqrt{a^2+b^2}>1$, see Eq. (\ref{eq:formP1P2}).) The parameter $\zeta$ has been defined in Eq. \eqref{eq:zeta}.

The following lemma is an important ingredient in our proof of decay estimates on $F_{Q_1,Q_2}(\tau)$.

\begin{lemma}\label{LM:q1q2}
Assume that conditions (I) and (II) in Lemma \ref{LM:est} hold.
Then the following two statements hold.
\begin{itemize}
\item[(a)] For $Q_1$ and $Q_2$ as in Eqs. \eqref{eq:defP1P2} and \eqref{eq:formP1P2}, we have that
    \begin{align}\label{eq:q1q2Ge}
    |Q_1|,\ |Q_2|>1,
 \end{align}
and it is ``almost true" that $|Q_1|\geq |Q_2|$ and $\zeta_0\geq \zeta$, in the sense that for some constant $C>0$,
\begin{align}\label{eq:alMzetaze}
\frac{|Q_2|-1}{|Q_1|-1},\ \frac{\zeta}{\zeta_0}\leq 1+C\rho_{0}^{\frac{1}{2}},
\end{align}
\item[(b)]
In the neighborhood $\rho\in [0,\frac{6}{5}\zeta_0]$, $\theta\in [0,R]$ of the critical point $(\rho,\theta)=(\zeta,0)$ we have that
\begin{align}\label{eq:lowerBG}
|G_{a,b}^{-2}(\rho,\theta,\alpha)|\lesssim \zeta_0^{-4}.
\end{align}
\end{itemize}
\end{lemma}
This lemma will be proven in Appendices \ref{sec:almostT} (statement ($a$)) and \ref{sec:lowerBG}
(statement ($b$)).

After identifying the critical points of the phase function $Y_{\sigma}$ in expression \eqref{eq:defY} and studying their neighborhoods, we decompose $F_{Q_1, Q_2}$ into four parts corresponding to the integration regions shown in Figure \ref{fig:digraph3} below, the parameters $\zeta,\ \zeta_0$ and $R$ having been introduced in Eqs. \eqref{eq:zeta}, \eqref{eq:zeta0} and \eqref{eq:defR}, respectively. These four contributions will be estimated in Lemmas \ref{LM:FR124} and \ref{LM:FR2} below.  Our estimates will then imply the desired bound on $F_{Q_1, Q_2}(\tau)$.
\begin{figure}[!htb]
\centering
\includegraphics[scale=1.1]{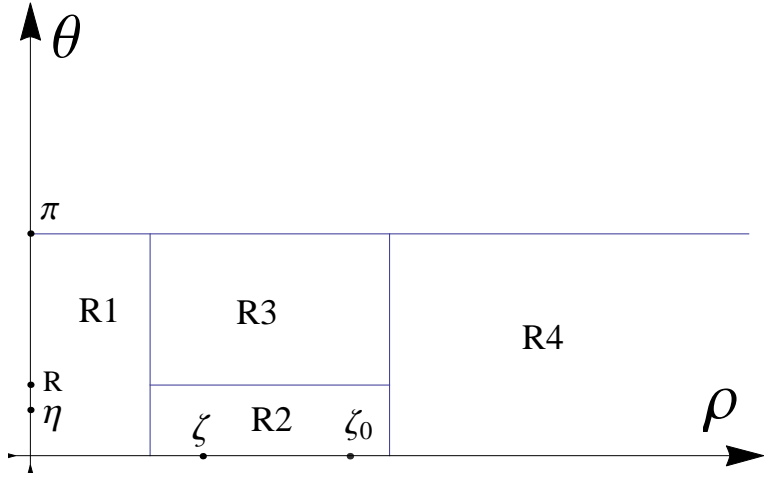}
\caption{Segmentation of Integration Regions}
\label{fig:digraph3}
\end{figure}

Corresponding to this partition of the domain of integration, the function $F_{Q_1, Q_2}$ introduced in \eqref{eq:defY} is split into four contributions,
\begin{align}\label{eq:segment}
F_{Q_1, Q_2}=F_{R1}+F_{R2}+F_{R3}+F_{R4},
\end{align}
with
\begin{align*}
F_{R1}:=&\int_{0}^{2\pi}\int_{0}^{\pi}\int_{0}^{\infty} H(\rho,\theta,\alpha, \tau)\ \chi(\frac{2\rho}{ \zeta })\ d\rho d\theta d\alpha\\
F_{R2}:=&\int_{0}^{2\pi}\int_{0}^{\pi}\int_{0}^{\infty}H(\rho,\theta,\alpha, \tau)\ \chi(\frac{6\theta}{5 R}) [1-\chi(\frac{2\rho}{ \zeta })]\chi(\frac{\rho}{ \zeta_0 })\  d\rho d\theta d\alpha\\
F_{R3}:=&\int_{0}^{2\pi}\int_{0}^{\pi}\int_{0}^{\infty} H(\rho,\theta,\alpha, \tau)\ [1-\chi(\frac{6\theta}{5 R})] [1-\chi(\frac{2\rho}{
\zeta })]\chi(\frac{\rho}{ \zeta_0 })\ d\rho d\theta d\alpha\\
F_{R4}:=&\int_{0}^{2\pi}\int_{0}^{\pi}\int_{0}^{\infty} H(\rho,\theta,\alpha, \tau)\ [1-\chi(\frac{\rho}{ \zeta_0 })]\ d\rho d\theta d\alpha
\end{align*}
where $$H(\rho,\theta,\alpha, \tau):=e^{-i\tau Y_{\sigma}(\rho,\theta)}[G_{a,b}(\rho,\theta,\alpha) -i0]^{-2}\frac{\rho^6 sin\theta}{1+\rho^2}g(\theta,\alpha)|\hat{V}(\rho)|^2, $$
and $\chi$ is a smooth cutoff function satisfying $\chi(x)=1$ if $x\leq \frac{11}{10}$ and $\chi(x)=0$ if $x\geq \frac{23}{20}.$

We now sketch the main ideas used in estimating $F_{Rk},\ k=1,2,3,4.$

The difficulties in estimating $F_{R1}, \ F_{R3}$ and $F_{R4}$ are connected to the fact that the denominator $G_{a,b}$ in the definition of $H(\rho, \theta, \alpha, \tau)$ vanishes at various points. To circumvent these difficulties we use the identity
\begin{align*}
[\rho G_{a,b}(\rho,\theta,\alpha)-i0]^{-2}=-\int_{0}^{\infty} du \int_u^{\infty} dz \text{  }e^{-i z [\rho \sqrt{1+\rho^2}-a \rho cos\theta -b \rho sin\theta cos\alpha]}\
\end{align*}
and find that
\begin{align}\label{eq:doubleInte}
F_{Rk}=\int_{0}^{\infty} d u \int_{u}^{\infty} dz \text{  }f_{Rk}(\tau,z)\ ,
\end{align}
where $f_{Rk}(\tau,z)$ is defined by
\begin{align}\label{fus. f_{Rk}}
f_{Rk}(\tau,z):=\int_{0}^{2\pi}d\alpha \int_{0}^{\pi}d\theta \int_{0}^{\infty} d\rho\text{   }e^{-i(\tau+z)X(\rho,\theta,\alpha)} \frac{\rho^8 sin\theta}{1+\rho^2} |\hat{V}(\rho)|^2 g(\theta,\alpha) \chi_{Rk}(\rho,\theta)\ d\rho d\theta d\alpha,
\end{align}
with $\chi_{Rk},\ k=1,3,4,$ cutoff functions, and
\begin{align}
X(\rho,\theta,\alpha):=\rho \sqrt{1+\rho^2}- [\frac{\tau}{\tau+z}\sigma+\frac{z}{\tau+z}a ]\rho cos\theta-\frac{z}{\tau+z} b \rho sin\theta cos\alpha.
\end{align}

In estimating $f_{R3}$, the key observation is that if, for an arbitrary, but fixed $\alpha$, the function $X$ does not have any critical points in the integration domain, and the critical points are located sufficiently far from the integration domain, then good decay estimates can be established.

In domains $R1$ and $R4$, the analysis (performed in the appendices, below) is really quite easy, because, after a certain transformation, two of the three integrals in the expressions for $f_{R1}$ and $f_{R4}$ can be evaluated in closed form, which facilitates the analysis.

The result of the analysis is summarized in the following lemma.
\begin{lemma}\label{LM:FR124}
\begin{align}
|f_{R1}|, |f_{R3}|, |f_{R4}|\lesssim (\tau+z)^{-\frac{7}{2}},
\end{align}
hence
\begin{align}
|F_{R1}|, |F_{R3}|, |F_{R4}|\lesssim \tau^{-\frac{3}{2}}
\end{align}
\end{lemma}

These bounds will be proven in Appendices \ref{sec:fR4}, \ref{sec:fR1} and \ref{sec:fR3}.

Among the four terms, $F_{R2}$ dominates, since a critical point is in the domain $R2$.

Using Lemma \ref{LM:est} and applying a standard stationary phase argument, one can prove the following lemma, (see Appendix \ref{sub:FR2}).
\begin{lemma}\label{LM:FR2}
\begin{align}
|F_{R2}|\lesssim \tau^{-\frac{3}{2}}.
\end{align}
\end{lemma}

\section{Proof of Lemma \ref{LM:est} when $|Q_2|\leq  1+ 10 \tau^{-\frac{2}{3}}$}\label{sec:A}
In this section, we suppose that
\begin{align}
1< |Q_2|\leq  1+ 10 \tau^{-\frac{2}{3}}.
\end{align}
(Note that, by condition (II) of Lemma \ref{LM:est}, we have that \ $|Q_2|>1$; see also \eqref{eq:q1q2Ge}.)
For certain technical reasons, this regime must be studied differently from the one corresponding to
$|Q_2|>  1+ 10 \tau^{-\frac{2}{3}}$. In the latter regime the stationary phase method is applicable because the condition $|Q_2|>  1+ 10 \tau^{-\frac{2}{3}}$ entails a separation of different critical points; see Eq. \eqref{eq:critical}. In the former regime, i.e., in the situation studied in this section, the critical points can be arbitrarily close to $(0,0).$ This forces us to make use of appropriate techniques, to be described below, to derive the desired estimates. However, these techniques cannot be used to understand the regime $|Q_2|>  1+ 10 \tau^{-\frac{2}{3}}$. (The reason is that the constant $c_1$ in \eqref{eq:expDecay} plays an adverse role and might become arbitrarily large -- see \eqref{eq:adverse}, below. This is related to what is called 'critical scaling'. We shall not elaborate on this point here.)

In the present case the main difficulty is that the function $G$ may vanish at several points. To overcome it we deform the contour of integration appropriately, namely from $\mathbb{R}_{+}$, to the curve $\Gamma$ shown in Figure \ref{fig:digraph2}, below; with the straight line part parameterized by
\begin{align}\label{eq:curve}
\rho=|\rho| e^{-i\gamma},\ \rho\in [0,15],\ \gamma \in (0,\frac{\pi}{6}].
\end{align}
The idea is motivated by arguments presented in~\cite{MR1681113, MR0495958}. The deformation of the integration contour used here is legitimate, because $\hat{V}$ and $\overline{\hat{V}}$ can be extended to functions analytic in a strip around the real axis. (Recall that we have assumed that $V$ decays exponentially fast.)

\begin{figure}[!htb]
\centering
\includegraphics[scale=1]{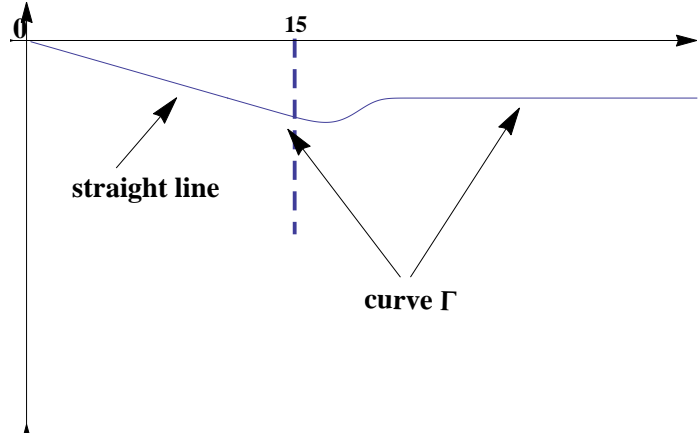}
\caption{Integration Contour for $\rho$}
\label{fig:digraph2}
\end{figure}

\begin{lemma}
For any $\rho\in \Gamma$, we have that $ImG\leq 0$, so that
 $$[G_{a,b}(\rho,\theta,\alpha)-i0]^{-2}=G^{-2}(\rho,\theta,\alpha),$$
and
\begin{align}\label{eq:rho4}
|G_{a,b}^{-2}(\rho,\theta,\alpha)|\lesssim |\rho|^{-4}.
\end{align}
\end{lemma}
The proof of this lemma is straightforward and is therefore omitted.

We conclude that the function $F_{Q_1, Q_2}$ takes the form
\begin{align}
F_{Q_1, Q_2}(\tau)=&\int_{0}^{2\pi}\int_{0}^{\pi}\int_{\Gamma} e^{-i\tau Y_{\sigma}(\rho,\theta)} G_{a,b}^{-2}(\rho,\theta,\alpha)\frac{\rho^6 sin\theta}{1+\rho^2}g(\theta,\alpha)|\hat{V}(\rho^2)|^2 \ d\rho d\theta d\alpha\nonumber\\
=&F_1+F_2, \label{eq:f1f2}
\end{align}
where
\begin{align*}
F_{1}:=&\int_{0}^{2\pi}\int_{0}^{\pi}\int_{\Gamma} e^{-i\tau Y_{\sigma}(\rho,\theta)}G_{a,b}^{-2}(\rho,\theta,\alpha)\frac{\rho^6 sin\theta}{1+\rho^2}g(\theta,\alpha)|\hat{V}(\rho^2)|^2 \chi(\rho)\ d\rho d\theta d\alpha\\
F_{2}:=&\int_{0}^{2\pi}\int_{0}^{\pi}\int_{\Gamma} e^{-i\tau Y_{\sigma}(\rho,\theta)}G_{a,b}^{-2}(\rho,\theta,\alpha)\frac{\rho^6 sin\theta}{1+\rho^2}g(\theta,\alpha)|\hat{V}(\rho^2)|^2 [1-\chi(\rho)]\ d\rho d\theta d\alpha
\end{align*} and $\chi$ is a smooth cutoff function satisfying $\chi(\rho)=1$ if $|\rho|\leq 12$ and $\chi(\rho)=0$ if $|\rho|\geq 14.$

It is easy to analyze $F_2$. We integrate by parts, using the identity
\begin{align*}
e^{-i\tau Y_{\sigma}}=\frac{1}{-i\tau \partial_{\rho}Y_{\sigma}} \partial_{\rho} e^{-i\tau Y_{\sigma}}.
\end{align*} A simple observation is that, for $|\rho|\geq 12,$
$|\frac{1}{\partial_{\rho}Y_{\sigma}}|\leq 1$. We may integrate by parts as many times as we wish and obtain that
\begin{align}\label{eq:estF2}
|F_{2}|\leq C(N)\tau^{-N},
\end{align}
for some finite constant $C(N)$ and for any $N> 0$.

We now turn to estimating $F_1$, the decay of which is caused by $e^{-i\tau Y_{\sigma}(\rho,\theta)}$.

By Eq. (\ref{eq:curve}), $\rho=|\rho|e^{-i\gamma}$, with $|\rho|\in [0,15]$ and
$0\leq \gamma \leq \frac{1}{6}\pi$.
It is convenient to introduce new variables, $r$ and $\beta$, by setting
\begin{align}
r:=\tau^{\frac{1}{3}}\rho,\ \beta:=\tau^{\frac{1}{3}}\theta.
\end{align}

\begin{lemma}\label{LM:exp}
For $|\rho|\leq 15$, we have that
\begin{align}\label{eq:expDecay}
-Im\ \tau Y_{\sigma}(\rho,\theta)\geq c_0 (\min\{r^2, r^3\}+r\beta^2)-c_1 r,
\end{align}
for some constants $c_0,\ c_1>0$.
\end{lemma}
This lemma will be proven in Subsection \ref{subsec:expde}.

To estimate $F_1$, we rescale the variables $r:=\rho \tau^{\frac{1}{3}}$ and $\beta= \theta \tau^{\frac{1}{3}}$ and then apply \eqref{eq:rho4} and \eqref{eq:expDecay} to obtain
\begin{align*}
|F_1|\lesssim \tau^{-\frac{5}{3}}\int_{0}^{2\pi}\int_{0}^{\tau^{\frac{1}{3}}\pi}\int_{0}^{\infty} e^{-c_0 (\min\{r^2,r^3\}+r\beta^2)+c_1 r} r^{2} \beta \frac{sin (\tau^{-\frac{1}{3}}\beta)}{\tau^{-\frac{1}{3}}\beta }\ dr d\beta d\alpha.
\end{align*}
This and the trivial bound $\frac{|sin\gamma|}{\gamma}\leq 1$, for any $\gamma\geq 0$, imply that
\begin{align}
|F_1|\lesssim \tau^{-\frac{5}{3}}.
\end{align}
With \eqref{eq:estF2} and \eqref{eq:f1f2} this then yields
\begin{align}
|F_{Q_1, Q_2}(\tau)|\lesssim \tau^{-\frac{5}{3}}.
\end{align}

To show that there is no singularity at $\tau=0$ we perform a direct estimate on \eqref{eq:f1f2} and find that it is bounded for $\tau\in [0, 1]$. Thus, for any $\tau\geq 0$,
\begin{align}
|F_{Q_1, Q_2}(\tau)|\lesssim (1+\tau)^{-\frac{5}{3}},
\end{align}
which is the desired estimate.

\subsection{Proof of Lemma \ref{LM:exp}}\label{subsec:expde}
We define a function $\tilde{Y}$ by
\begin{align}\label{eq:tildeY}
\tilde{Y}(r,\beta):=\tau Y_{\sigma}(\rho,\theta)=\frac{r^3}{1+\sqrt{1+\tau^{-\frac{2}{3}}r^2}} +r\frac{1-cos(\tau^{-\frac{2}{3}} \beta)}{\tau^{-\frac{2}{3}}}+ \frac{1-\sigma}{\tau^{-\frac{2}{3}}} r cos(\tau^{-\frac{1}{3}}\beta),
\end{align}
with $r$ and $\beta$ defined by $\rho=\tau^{-\frac{1}{3}} r$ and $\theta=\tau^{-\frac{1}{3}}\beta.$
The three terms on the right side of \eqref{eq:tildeY} are studied separately.

Recall that $r=|r|e^{-i\gamma}$, for some $\gamma\in (0,\frac{\pi}{6}]$; (see \eqref{eq:curve}). With the observation that
$$\frac{1}{1+\sqrt{1+\tau^{-\frac{2}{3}}r^2}}=
|\frac{1}{1+\sqrt{1+\tau^{-\frac{2}{3}}r^2}}|e^{i\tilde\gamma},$$ with $\tilde\gamma\in (0,\gamma)$,
this implies that
\begin{align}
Im[\frac{r^3}{1+\sqrt{1+\tau^{-\frac{2}{3}}r^2}}]=-|\frac{r^3}{1+\sqrt{1+\tau^{-\frac{2}{3}}r^2}}|\ sin(3\gamma- \tilde\gamma)\leq -\frac{1}{3} \min\{|r|^2, |r|^3\} \ sin(2\gamma),
\end{align}
where we have used that $3\gamma -\tilde\gamma\geq 2\gamma$ and $\gamma\in (0,\frac{\pi}{6})$.

To bound the second term on the right side of \eqref{eq:tildeY} we observe that, for any
$\beta\in [0,\pi t^{\frac{1}{3}}]$, there exists a constant $c_0>0$ such that
$$
\frac{1-cos(\tau^{-\frac{1}{3}} \beta)}{\tau^{-\frac{2}{3}}}\geq c_0 \beta^2.
$$
This shows that
\begin{align}
Im [r\frac{1-cos(\tau^{-\frac{1}{3}} \beta)}{\tau^{-\frac{2}{3}}}]=-sin\gamma \ |r|\frac{1-cos(\tau^{-\frac{1}{3}} \beta)}{\tau^{-\frac{2}{3}}}\leq -c_0\sin\gamma\ |r|\beta^2.
\end{align}

Only the last term in \eqref{eq:tildeY} may have a positive imaginary part in the regime where
$(\sigma-1)\tau^{\frac{2}{3}}= (|Q_2|-1)\tau^{\frac{2}{3}}\leq 10$:
\begin{align}\label{eq:adverse}
Im\frac{1-\sigma}{\tau^{-\frac{2}{3}}} r cos(\tau^{-\frac{1}{3}}\beta)=-\frac{1-\sigma}{\tau^{-\frac{2}{3}}}|r| sin\gamma \ cos(\tau^{-\frac{1}{3}}\beta)\leq  10 |r|\ sin\gamma.
\end{align}

Collecting the above estimates and using the fact that $\gamma\in (0,\frac{\pi}{6}]$ we arrive at the desired estimates.

\section{Proof of Lemma \ref{LM:HInfty}}\label{sec:Hinf}
In what follows we only prove \eqref{eq:singular}. The proofs of the other two bounds in
Lemma \ref{LM:HInfty} are easier, thanks to either the presence of a gradient $\nabla_{x}$, (which, in Fourier space, yields an additional factor of a momentum $k$), or to the absence of a singularity in $H_{\infty}^{-1}$, respectively.

After diagonalizing the operator $H_0$, as in \eqref{eq:diag}, one finds that
\begin{align}
e^{tH_{\infty}} H_{\infty}^{-1} \left[
\begin{array}{ccc}
0\\
W^{Y_{0}}
\end{array}
\right]=&-Re\left[
\begin{array}{cc}
\frac{\sqrt{-\Delta}}{2\sqrt{-\Delta+1}} e^{it(L-iP_{\infty}\cdot \partial_{x})} [L-iP_\infty\cdot\nabla_{x}]^{-1}W^{Y_0}\\
\frac{i}{2} e^{it(L-iP_{\infty}\cdot \partial_{x})} [L-iP_\infty\cdot\nabla_{x}]^{-1}W^{Y_0}.
\end{array}
\right]
=:-Re\left[
\begin{array}{cc}
\Psi_1\\
\Psi_2
\end{array}
\right]. \nonumber
\end{align}

In what follows we only study $\Psi_2$. (The study of $\Psi_1$ is easier, thanks to the presence of
$\sqrt{-\Delta}$.)

Fourier transformation yields
\begin{align*}
\Psi_2(t)
=C\int_{\mathbb{R}^3}  e^{it [|k|\sqrt{1+|k|^2}-\mu\cdot k]}[|k|\sqrt{1+|k|^2}-P_\infty \cdot k]^{-1}  |k|^{\frac{3}{2}}\hat{V}(k)\ dk,
\end{align*}
where $\mu$ is defined as $\mu:=\frac{-x+Y_0-tP_{\infty}}{t}\in \mathbb{R}^3.$
After rotating the coordinate axes we can assume that
\begin{align}
\mu= (\sigma,0,0)\ \text{and}\ P_{\infty}= (p_1,p_2,0).
\end{align}

Introducing polar coordinates one arrives at the expression
\begin{align}
\Psi_2(t)
=& \int_{0}^{2\pi}\int_{0}^{\pi}\int_{0}^{\infty} e^{it\rho[ \sqrt{1+\rho^2}-\sigma cos\theta]}\frac{\rho^{\frac{5}{2}}sin\theta}{\sqrt{1+\rho^2}-p_1 cos\theta-p_2 sin\theta cos\alpha} f(\theta,\alpha)\hat{V}(\rho) \ d\rho d\theta d\alpha\nonumber\\
=&\int_{0}^{2\pi}\int_{0}^{\pi}\int_{0}^{\infty} e^{it\rho[ \sqrt{1+\rho^2}-\sigma cos\theta]}\frac{\rho^{\frac{5}{2}}sin\theta}{\sqrt{1+\rho^2}-q cos(\theta-\beta)} f(\theta,\alpha)\hat{V}(\rho) \ d\rho d\theta d\alpha, \label{eq:estD2t}
\end{align}
where $f(\theta,\alpha)$ is a polynomial in $sines$ and $cosines$ of $\theta$ and $\alpha$.
In the last step we have rewritten $$p_1 cos\theta+p_2 sin\theta cos\alpha=q cos(\theta-\beta),$$
with $q:=\sqrt{p_1^2+p_2^2 cos^2\alpha}\leq |P_{\infty}|=1$, and $\beta\in [0,2\pi)$ satisfying
$cos\beta=\frac{p_1}{q}$, $sin\beta=\frac{p_2 cos\alpha}{q}.$

The desired estimate will follow after integrating by parts. There is only one minor difficulty: one must show that the singularity produced by $\frac{1}{\sqrt{1+\rho^2}-q cos(\theta-\beta)}$ and derivatives thereof is compensated by other factors, so as to yield an integrable integrand on the right side of \eqref{eq:estD2t}.

If $|\sigma|\geq \frac{1}{2}$, we integrate by parts in $\theta$, using
\begin{align*}
sin\theta \rho e^{-it \rho cos\theta} =\frac{1}{it}\partial_{\theta}e^{-i t \rho cos\theta},
\end{align*}
 to find that
\begin{align*}
\Psi_2(t)=&-\frac{1}{it} \int_{0}^{2\pi}\int_{0}^{\pi}\int_{0}^{\infty} e^{it\rho[ \sqrt{1+\rho^2}-\sigma cos\theta]} \rho^{\frac{3}{2}} f(\theta,\alpha) \frac{q sin(\theta-\beta)}{[\sqrt{1+\rho^2}-q cos(\theta-\beta)]^2} \hat{V}(\rho) d\rho d\theta d\alpha+\cdots\\
=:& \tilde{\Psi}_{2}(t)+\cdots,
\end{align*}
(recall that $\beta$ is independent of $\theta$). The contributions not displayed explicitly are easier to control, because the integrand is less singular. We will not study them.

In order to control the denominator in the integrand of the term defining $\tilde{\Psi}_2$, we use that
\begin{align*}
\sqrt{1+\rho^2}-q cos(\theta-\beta)\gtrsim \rho^2+(\theta-\beta)^2,
\end{align*}
using that $0\leq q\leq 1$.
Hence
\begin{align}
\rho^{\frac{3}{2}}|\frac{ sin(\theta-\beta)}{[\sqrt{1+\rho^2}-q cos(\theta-\beta)]^2}|\lesssim  \frac{\rho^{\frac{3}{2}}}{[\rho^2+(\theta-\beta)^2]^{\frac{3}{2}}}
\end{align}
is integrable in the variables $\rho$ and $\theta$ on the open set $\rho^2+(\theta-\beta)^2\leq 1.$
This and the fast decay of $\hat{V}(\rho)$ imply that the integral is finite, and hence
\begin{align}\label{eq:d2sigma}
|\Psi_2(t)|,\ |\tilde{\Psi}_2(t)|\lesssim t^{-1}.
\end{align}

Next, we consider the case where $|\sigma|<\frac{1}{2}$. We integrate by parts, using the identity
\begin{align*}
e^{it\rho[ \sqrt{1+\rho^2}-\sigma cos\theta]}
=\frac{1}{it [\sqrt{1+\rho^2}+\frac{\rho^2}{\sqrt{1+\rho^2}}-\sigma cos\theta]}\partial_{\rho}
e^{it\rho[ \sqrt{1+\rho^2}-\sigma cos\theta]}.
\end{align*}
The denominator is bounded away from 0,
\begin{align*}
\sqrt{1+\rho^2}+\frac{\rho^2}{\sqrt{1+\rho^2}}-\sigma cos\theta\geq \frac{1}{2},
\end{align*}
and we finally get that
\begin{align}
|\Psi_2(t)|\lesssim t^{-1}.
\end{align}

This bound, together with \eqref{eq:d2sigma}, implies that, for an arbitrary $\sigma\in \mathbb{R}$, \begin{align}
|\Psi_2(t)|\lesssim t^{-1}.
\end{align}

Starting from the expression in \eqref{eq:estD2t}, one easily sees that $\Psi_2(t)$ is bounded uniformly, for $|t|\leq 1.$

This completes our proof of \eqref{eq:singular}.

\appendix

\section{Proof of Statement (a) of Lemma \ref{LM:q1q2}}\label{sec:almostT}
To prepare the ground for later analysis we prove a result somewhat stronger than Statement (a) in Lemma \ref{LM:q1q2}. For the convenience of the reader we first repeat some definitions: The parameters $\sigma, a$ and $b$ are defined by setting
\begin{align*}
Q_1=(a,b,0),\ Q_2=(\sigma,0,0),
\end{align*}
with $Q_1$ and $Q_2$ given by
\begin{align*}
Q_1:=P_s\ \text{and}\ Q_2:=\frac{1}{t-s} \int_{s}^{t} P_{s_1}\ ds_1;
\end{align*}
see Eqs. \eqref{eq:formP1P2} and \eqref{eq:defP1P2}, respectively. (Here we are using the spherical symmetry of $W$ to turn $Q_1$ and $Q_2$ into special directions.) A parameter $\zeta_0$ has been introduced as the solution of Eq. \eqref{eq:zeta0}:
\begin{align*}
\sqrt{1+\zeta_0^2}+\frac{\zeta_0^2}{\sqrt{1+\zeta_0^2}}=\sqrt{a^2+b^2}=|Q_2|>1.
\end{align*}
Finally, a parameter $\zeta$ has been defined as the solution of Eq. \eqref{eq:zeta}, viz.
\begin{align*}
\sqrt{1+\zeta^2}+
\frac{\zeta^2}{\sqrt{1+\zeta^2}}=\sigma \text{    }\text{and}\ \ cos(\eta)=\frac{1}{\sigma}.
\end{align*}

\begin{proposition}\label{prop:decayP1P2}
If conditions (I) and (II) of Lemma \ref{LM:est} hold then
\begin{align}\label{eq:almostCL}
a,\ \sigma>1,\ \text{and}\ |b|\lesssim \rho_{0}^{\frac{1}{2}} (\sqrt{a^2+b^2}-1).
\end{align}
Moreover,
\begin{align}\label{eq:almostDec}
\frac{\zeta}{\zeta_0},\ \frac{\sigma-1}{\sqrt{a^2+b^2}-1}\leq 1+O( \rho_{0}^{\frac{1}{2}}).
\end{align}
\end{proposition}

\begin{proof}

The first inequality in condition (II) of Lemma \ref{LM:est} implies that
$$|P_t| \geq 1+\rho_{0}^{\frac{1}{4}}(1+\rho_0 t)^{-\frac{2}{5}}.$$
We temporarily change our coordinates such that $P_t = (|P_t|,0,0)$. Given an arbitrary time $s_1\in [s,t]$, we let $q_1(s_1),\ q_2(s_1), q_3(s_1)$ denote the three components of the momentum vector $P_{s_{1}}$.
The second inequality in condition (II) of Lemma \ref{LM:est} and the fact that $|P_{s_1}|$ is bounded then imply that
\begin{align}\label{eq:inside1}
|q_2(s_1)|,\ |q_3(s_1)|\lesssim \rho_{0}^{\frac{3}{4}} (1+\rho_0 s_1)^{-\frac{2}{5}}.
\end{align}
These bounds and the first inequality in condition (II) of Lemma \ref{LM:est}, i.e.,
$$|P_{s_1}|\geq 1+ \rho_{0}^{\frac{1}{4}} (1+\rho_0 s_1)^{-\frac{2}{5}}$$
then imply that
\begin{align}\label{eq:inside2}
q_1(s_1)\geq 1+ \rho_{0}^{\frac{1}{4}} (1+\rho_0 s_1)^{-\frac{2}{5}}(1-O(\rho_{0}^{\frac{1}{2}}))
\end{align}

Consequently, defining $\tilde{p}_1, \tilde{p}_2$ and $\tilde{p}_3$ by
\begin{align}\label{eq:int3d}
\frac{1}{t-s}\int_{s}^{t} P_{s_1} ds_1=: (\tilde{p}_1,\tilde{p}_2,\tilde{p}_3),
\end{align}
we have that
\begin{align*}
|\tilde{p}_2|,\ |\tilde{p}_3| \leq &\frac{\rho_{0}^{\frac{3}{4}}}{t-s} \int_{s}^{t} (1+\rho_0 s_1)^{-\frac{2}{5}} \ ds_1,
\end{align*}
and
\begin{align*}
\tilde{p}_1\geq & 1+\frac{\rho_{0}^{\frac{1}{4}}}{t-s} \int_{s}^{t} (1+\rho_0 s_1)^{-\frac{2}{5}} \ ds_1 \ (1-O(\rho_{0}^{\frac{1}{2}})).
\end{align*}

Next, we consider $P_s$. By setting $s_1=s$ in \eqref{eq:inside1} and \eqref{eq:inside2} and recalling that $q_1(s), q_2(s)$ and $q_3(s)$ are the components of $P_s$, we find that
\begin{align}
|q_2(s)|,\ |q_3(s)|\lesssim & \rho_{0}^{\frac{3}{4}}(1+\rho_0 s)^{-\frac{2}{5}},\nonumber\\
q_1(s)\geq & 1+\rho_0^{\frac{1}{4}}(1+\rho_0 s)^{-\frac{2}{5}}(1-O(\rho_{0}^{\frac{1}{2}})).\label{eq:q123}
\end{align}

A space rotation by an angle of order
$ \frac{\rho_{0}^{\frac{3}{4}}}{t-s} \int_{s}^{t} (1+\rho_0 s_1)^{-\frac{2}{5}} \ ds_1\leq \rho_{0}^{\frac{3}{4}}(1+\rho_0 s)^{-\frac{2}{5}}$ brings  $\frac{1}{t-s}\int_{s}^{t} P_{s_1}\ ds_1$ and $P_s$ back into their original positions, i.e.,
$$\frac{1}{t-s} \int_{s}^{t} P_{s_1}\ ds_1=(\sigma,0,0)\ \text{and }\ P_s=(a,b,0), $$
for some $\sigma>0.$
The above analysis then implies that
\begin{align}\label{eq:abgeq1}
\sigma>1,\ a\geq 1+\rho_0^{\frac{1}{4}}(1+\rho_0 s)^{-\frac{2}{5}}(1-O(\rho_{0}^{\frac{1}{2}}))\ \text{and}\ |b|\lesssim \rho_{0}^{\frac{3}{4}}(1+\rho_0 s)^{-\frac{2}{5}}
\end{align} which implies the desired estimates in \eqref{eq:almostCL}.

In the remainder of this appendix we prove the inequalities in \eqref{eq:almostDec}. For this purpose, we introduce two regimes, $s\geq \rho_0^{-\frac{1}{10}}$ and $s<\rho_0^{-\frac{1}{10}}$, that will be studied separately.

For $s\geq \rho_0^{-\frac{1}{10}}$, the fact that $\frac{d}{dt}|P_t|\leq 0$, for any $t\geq \rho_0^{-\frac{1}{10}}$, (see condition (I) in Lemma \ref{LM:est}), implies that
\begin{align}
\sigma=|\frac{1}{t-s}\int_{s}^{t} P_{s_1} ds_1|\leq |P_s|=\sqrt{a^2+b^2},
\end{align}
which implies the second inequality in \eqref{eq:almostDec}.

For $s<\rho_0^{-\frac{1}{10}}$, Proposition \ref{prop:smallVar} and our choice of initial conditions, in particular $|P_0|\geq \frac{11}{10}$, imply that
\begin{align}\label{eq:ageq1110}
 |P_{s_1}|\geq \max_{s\in [0,\rho_0^{-\frac{1}{10}}]}|P_{s}|-O(\rho_0^{\frac{9}{10}})\geq \frac{11}{10}-O(\rho_0^{\frac{9}{10}}),
\end{align}
for any $s_{1}\in [0,\rho_0^{-\frac{1}{10}}]$.
Using now that $\frac{d}{dt}|P_t|\leq 0$, for any $t\geq \rho_0^{-\frac{1}{10}}$, we conclude that
\begin{align*}
\sigma-1=|\frac{1}{t-s} \int_{s}^{t} P_{s_1}\ ds_1|-1\leq& \displaystyle\max_{s\in [0,\rho_0^{-\frac{1}{10}}]}|P_s|-1\\
\leq &|P_s|-1+O(\rho_0^{\frac{9}{10}})\\
=&\sqrt{a^2+b^2}-1+O(\rho_{0}^{\frac{9}{10}}).
\end{align*}
Dividing both sides by $\sqrt{a^2+b^2}-1>0$ and then using that
$\sqrt{a^2+b^2}-1\geq \frac{1}{10}-O(\rho_0^{\frac{9}{10}})$, we arrive at the second inequality in \eqref{eq:almostDec}.

To prove the first inequality in \eqref{eq:almostDec} we rewrite the expressions for $\zeta$ and $\zeta_0$ (see \eqref{eq:zeta} and \eqref{eq:zeta0}) as follows: We introduce a function
$\zeta(\eta)>0$ as the solution of the equation
\begin{align}
\frac{\zeta^2(\eta)}{1+\sqrt{1+\zeta^2(\eta)}}+\frac{\zeta^2(\eta)}{\sqrt{1+\zeta^2(\eta)}}=&\eta^2
\end{align}
and set
$$\zeta=\zeta(\sqrt{\sigma-1})\ \text{and}\ \zeta_0=\zeta(\sqrt{\sqrt{a^2+b^2}-1}).$$
The first inequality in \eqref{eq:almostDec}, hence estimate \eqref{eq:alMzetaze}, then follows from the observations that (i) $\zeta$ is a continuous increasing function of $\eta$, and (ii)  $\sqrt{\sigma-1}$ is less than $\sqrt{\sqrt{a^2+b^2}-1}$, up to an error term that tends to 0, as $\rho_0 \searrow 0$; (see \eqref{eq:almostDec}).
\end{proof}

\section{Proof of Statement (b) of Lemma \ref{LM:q1q2}}\label{sec:lowerBG}

We start our considerations by simplifying the problem.

By definition,
$$-G_{a,b}(\rho,\theta,\alpha)=-\sqrt{1+\rho^2}+a cos\theta+ b sin\theta cos\alpha.$$
As shown in \eqref{eq:almostCL}, the parameter $b$ is very small (for small values of $\rho_0$).
Hence the function $-G_{a,b}$ is decreasing in the variables $\rho$ and $\theta$, for $\theta\in [0,\frac{\pi}{4}].$ We therefore have that
$$-G_{a,b}(\rho,\theta,\alpha)|_{\rho\leq \frac{6}{5}\zeta_0,\ \theta\leq R}\geq -G_{a,b}(\rho,\theta,\alpha)|_{\rho= \frac{6}{5}\zeta_0,\ \theta= R},$$
and we recall that $R=\frac{1}{5}\frac{\zeta_0}{\sqrt{1+\zeta_0^2}}\leq \frac{1}{5}$.
Consequently, in order to prove statement ($b$) of Lemma \ref{LM:q1q2}, it suffices to show that
\begin{align}
-G_{a,b}(\rho,\theta,\alpha)|_{\rho= \frac{6}{5}\zeta_0,\ \theta= R}\gtrsim \zeta_0^2.
\end{align}
This follows from a straightforward computation:
\begin{align}
-G_{a,b}(\frac{6}{5}\zeta_0,R,\alpha)=& 1 -\sqrt{1+(\frac{6}{5}\zeta_0)^2}+a-1+ a (cosR-1)+b sinR \ cos\alpha\nonumber\\
=&-\frac{(\frac{6}{5}\zeta_0)^2}{1+\sqrt{1+(\frac{6}{5}\zeta_0)^2}}+a-1+ a (cosR-1)+b sinR\  cos\alpha\nonumber\\
\geq & -\frac{6}{5} \frac{\zeta_0^2}{1+\sqrt{1+\zeta_0^2}}+[\sqrt{a^2+b^2}-1][1-C\rho_{0}^{\frac{1}{2}}]+ a (cosR-1)\nonumber\\
=&-\frac{1}{5}\frac{\zeta_0^2}{1+\sqrt{1+\zeta_0^2}}[1+5C\rho_{0}^{\frac{1}{2}}]+
\frac{\zeta_0^2}{\sqrt{1+\zeta_0^2}}[1-C\rho_{0}^{\frac{1}{2}}]+ a (cosR-1)\nonumber\\
\geq & \frac{3}{4}\frac{\zeta_0^2}{\sqrt{1+\zeta_0^2}}+a (cosR-1)\nonumber\\
\geq & \frac{1}{2}\frac{\zeta_0^2}{\sqrt{1+\zeta_0^2}}
\label{eq:thirdLine}
\end{align}
Four facts have been used here: (1) In the third step, we have used that $|b|\lesssim (\sqrt{a^2+b^2}-1) \rho_{0}^{\frac{1}{2}}$ and $a>1$ (see Proposition \ref{prop:decayP1P2}), which implies that
$$|b sinR cos\alpha|+|\sqrt{a^2+b^2}-a|\leq C\rho_{0}^{\frac{1}{2}}[\sqrt{a^2+b^2}-1],$$
for some constant $C>0,$ (recall that $\rho_0$ is chosen small enough); (2) in the fourth step, we have used the identity
$$\frac{\zeta_0^2}{1+\sqrt{1+\zeta_0^2}}+\frac{\zeta_0^2}{\sqrt{1+\zeta_0^2}}=\sqrt{a^2+b^2}-1,$$
see \eqref{eq:zeta0};
(3) in the third last step we have used that $\frac{\zeta_0^2}{1+\sqrt{1+\zeta_0^2}}\leq \frac{\zeta_0^2}{\sqrt{1+\zeta_0^2}}$;
 and (4) in the second but last step we have used the smallness of
 $R=\frac{1}{5}\frac{\zeta_0}{\sqrt{1+\zeta_0^2}}\leq \frac{1}{5}$
 to find that
$$cos R-1=-\frac{1}{2} sin^2(\frac{R}{2})\geq -R^2=-\frac{1}{25}\frac{\zeta_0^2}{1+\zeta_0^2},$$
and the fact that $3\sqrt{1+\zeta_0^2}\geq a\geq 1$ has been used, (see \eqref{eq:almostCL}).

\section{Bound on the function $f_{R4}$ defined in \eqref{fus. f_{Rk}}}\label{sec:fR4}

The function $f_{R4}$ (see  \eqref{fus. f_{Rk}}) appears as the integrand of one contribution, denoted by
$F_{R4}$ (see \eqref{eq:segment}), to the function $F_{Q_1, Q_2}$ given in Eq. \eqref{eq:defY}, which can equivalently be expressed as in Eq. \eqref{eq:defFp1p2}. If we do not introduce special coordinates then $f_{R4}$ can be expressed as the matrix-valued function (also denoted by $f_{R4}$) given by
$$f_{R4}(\tau, z):=\langle \nabla_{x}V ,\ \frac{(-\Delta)^2}{\sqrt{-\Delta+1}} e^{i(L-i\frac{z}{\tau+z}Q_1\cdot \partial_x-i\frac{\tau}{\tau+z}Q_2\cdot \partial_x )(\tau+z)}\ [1-\chi(\frac{\sqrt{-\Delta}}{\zeta_0})]\ \nabla_{x}V\rangle.$$

It is somewhat disagreeable that the direction of the vectors $\frac{z}{\tau+z}Q_1+\frac{\tau}{\tau+z}Q_2$ may depend on $s$ and $t$. The presence of $\nabla_{x}$ in the above expression for
$f_{R4}(\tau,z)$ makes it plain that $f_{R4}(\tau,z)$ is  a $3\times 3$ matrix-valued function. By conjugating $f_{R4}(\tau,z)$ with a suitably chosen rotation, $M=M(\frac{\tau}{\tau+z}, Q_1, Q_2)$,
\begin{align}
f_{R4}(\tau,z)=M\ \Gamma(\tau,z)\ M^{T},
\end{align}
we may achieve that $\Gamma$ takes the form
$$\Gamma(\tau,z):=\langle \nabla_{x}V ,\ \frac{(-\Delta)^2}{\sqrt{-\Delta+1}} e^{i(L-iq\partial_{x_3} )(\tau+z)}\ [1-\chi(\frac{\sqrt{-\Delta}}{\zeta_0})]\ \nabla_{x}V\rangle,$$
where
\begin{align}\label{eq:defQ}
q:=|\frac{z}{\tau+z}Q_1+\frac{\tau}{\tau+z} Q_2|=\sqrt{[\frac{\tau}{\tau+z}\sigma+\frac{z}{\tau+z}a]^2+[\frac{z}{\tau+z}b]^2}.
\end{align}
From now on we study $\Gamma.$

Because $V$ has been assumed to be spherically symmetric, only the diagonal elements of
$\Gamma(\tau,z)$ can be non-zero.
These diagonal elements can be expressed in terms of only two functions, which, after Fourier transformation and introduction of polar coordinates, are seen to be given by:
\begin{align*}
\Gamma_1(\tau, z):= &\int_0^{\pi}\int_{0}^{\infty} \frac{\rho^8 sin\theta}{1+\rho^2} e^{i(\tau+z)[\rho\sqrt{1+\rho^2}-q \rho cos\theta]  }\  [1-\chi(\frac{\rho}{ \zeta_0 })] cos^2\theta\ |\hat{V}(\rho)|^2\ d\rho d\theta\\
\Gamma_2(\tau, z):= &\int_0^{\pi}\int_{0}^{\infty} \frac{\rho^8 sin\theta}{1+\rho^2} e^{i(\tau+z)[\rho\sqrt{1+\rho^2}-q \rho cos\theta]  }\  [1-\chi(\frac{\rho}{ \zeta_0 })] sin^2\theta\ |\hat{V}(\rho)|^2\ d\rho d\theta.
\end{align*}
(Only double integrals, instead of triple integrals, appear in these expressions, because one variable,
$\alpha$, has been integrated out.)

We proceed to estimating the function $\Gamma_{1}(\tau, z)$; (similar arguments can then be applied to estimating $\Gamma_2(\tau, z)$).

Integrating by parts, using the identity
\begin{align}\label{eq:InPartsSin}
sin\theta\ e^{-i(\tau+z) q \rho\ cos\theta}=\frac{1}{i(\tau+z)q\rho } \partial_{\theta} e^{-i(\tau+z) q \rho\ cos\theta},
\end{align}
one obtains that
\begin{align}\label{eq:M123}
\Gamma_1(\tau, z)=M_1+M_2+M_3
\end{align}
where
 $$M_1:=-\frac{1}{i(\tau+z)q} \int_{0}^{\infty} \frac{\rho^7}{1+\rho^2} e^{-i(\tau+z) [\rho\sqrt{1+\rho^2}-q \rho]}[1-\chi(\frac{\rho}{ \zeta_0 })] \ |\hat{V}(\rho)|^2\ d\rho,$$
$$M_2:=\frac{1}{i(\tau+z)q} \int_{0}^{\infty} \frac{\rho^7}{1+\rho^2} e^{-i(\tau+z) [\rho\sqrt{1+\rho^2}+q \rho]}[1-\chi(\frac{\rho}{ \zeta_0 })] \ |\hat{V}(\rho)|^2\ d\rho,$$
and
$$M_3:=\frac{2}{i(\tau+z)q}\int_0^{\pi}\int_{0}^{\infty} \frac{\rho^7 sin\theta}{1+\rho^2} e^{-i(\tau+z)[\rho\sqrt{1+\rho^2}-q \rho cos\theta]  }\  [1-\chi(\frac{\rho}{ \zeta_0 })] cos\theta\ |\hat{V}(\rho)|^2
\ d\rho d\theta.$$
Here $M_1$ and $M_2$ are boundary terms arising when integrating by parts.

We claim that
\begin{align}\label{eq:Mk123}
|M_{k}|\lesssim (\tau+z)^{-\frac{11}{3}},\ k=1,2,3,
\end{align}
which obviously implies the desired bounds on $\Gamma_1$, and hence on $f_{R4}$.

We now study $M_1$ and $M_2$ in detail.
(The term $M_3$ is analyzed by integrating by parts twice, using \eqref{eq:InPartsSin}, which converts it into a sum of terms similar to $M_1$ and $M_2$. We omit details.)

The idea underlying our treatment of $M_1$ has been sketched after Remark \ref{re:enoughZero}, Section \ref{sec:MainTHM}, assuming that $|Q_1|-1=O(\zeta_0^2)$ is small. We change variables by setting $\rho=:\zeta_0 r$ and obtain that
\begin{align}\label{eq:estM1}
M_1=\frac{\zeta_0^{8}}{i(\tau+z)q} \ \int_{0}^{\infty} \frac{r^7}{1+(\zeta_0 r)^2} e^{-i(\tau+z) \zeta_0^{3} X_{\zeta_0}(r)}[1-\chi(r)] \ |\hat{V}(\zeta_0 r)|^2\ dr,
\end{align}
where
\begin{align}\label{eq:defXr}
X_{\zeta_0}(r):=  \frac{\sqrt{1+(\zeta_0 r)^2}-1}{\zeta_0^2}r +\frac{1-q}{\zeta_0^2} r.
\end{align}
To exhibit the desired decay we integrate by parts, using the identity
$$e^{-i(\tau+z) \zeta_0^{3} X_{\zeta_0}(r)}= -\frac{1}{i(\tau+z)\zeta_0^3} \frac{1}{\partial_{r}X_{\zeta_0}(r)}\partial_{r}e^{-i(\tau+z) \zeta_0^{3} X_{\zeta_0}(r)}.$$
The denominator is controlled by observing that, on the support of the cutoff function $[1-\chi(r)]$, the function $X_{\zeta_0}(r)$ does not have any critical points and
\begin{align}\label{eq:estTiD1}
\frac{1}{|\partial_{r} X_{\zeta_0}(r)|}\lesssim \frac{1}{1+r}.
\end{align}
Supposing that \eqref{eq:estTiD1} holds, one may integrate by parts as many times as one wishes without producing boundary terms, thanks to the presence of the cutoff function. This leads to the bound
\begin{align}\label{eq:estD100}
|M_{1}|\leq C_{N}\frac{\zeta_0^{8}}{\tau+z} [(\tau+z)\zeta_0^3]^{-N}.
\end{align}

We now must estimate $\zeta_0$. For small values of $\sigma-1$ we have that
$\zeta=O(\sqrt{\sigma-1}),$ as follows from \eqref{eq:zeta}. We recall that we are considering the regime where $10 \tau^{-\frac{2}{3}}\leq |Q_2|-1=\sigma-1$. It has been shown in Proposition \ref{prop:decayP1P2} that
\begin{align}\label{eq:zetaZeta0}
\zeta_0\gtrsim  \zeta\gtrsim \tau^{-\frac{1}{3}}.
\end{align}
Plugging this into \eqref{eq:estD100}, we arrive at the desired bound  \eqref{eq:Mk123} on $|M_1|$.

Estimating $M_2$ is significantly easier. We integrate by parts, using
\begin{align}
e^{-i(\tau+z) [\rho\sqrt{1+\rho^2}+q \rho]}=-\frac{1}{i(\tau+z)\partial_{\rho}[\rho\sqrt{1+\rho^2}+q \rho]}\partial_{\rho}e^{-i(\tau+z) [\rho\sqrt{1+\rho^2}+q \rho]}
\end{align}
To control the denominator we use the fact that $q>0$ and find that
\begin{align*}
\partial_{\rho} [\rho\sqrt{1+\rho^2}+q \rho]=\sqrt{1+\rho^2}+\frac{\rho^2}{\sqrt{1+\rho^2}}+q\geq \sqrt{1+\rho^2}.
\end{align*}
This allows us to integrate by parts as many times as we wish, with the result that
\begin{align}
|M_2|\leq C_N (\tau+z)^{-N},
\end{align}
for any $N\in \mathbb{N}$, which implies the desired bound \eqref{eq:Mk123}.

To complete the proof we finally show that \eqref{eq:estTiD1} holds. A direct computation shows that
\begin{align}
\partial_{r}X_{\zeta_0}(r) &=\frac{r^2}{1+\sqrt{1+(\zeta_0 r)^2}}+\frac{r^2}{\sqrt{1+(\zeta_0 r)^2}}+\frac{1-q}{\zeta_0^2}\nonumber\\
&=\Phi_1+\Phi_{2}\label{eq:D112},
\end{align}
with $$\Phi_{1}(r):=\frac{r^2}{1+\sqrt{1+(\zeta_0 r)^2}}+\frac{r^2}{\sqrt{1+(\zeta_0 r)^2}}+\frac{1-\sqrt{a^2+b^2}}{\zeta_0^2}$$ and $$\Phi_{2}:=\frac{\sqrt{a^2+b^2}-q}{\zeta_0^2}.$$
The two facts,
(i) $\Phi_{1}(r=1)=0$, which follows from definition \eqref{eq:zeta0} of $\zeta_0$, and
(ii) the support of the cutoff function $1-\chi(r)$ is contained in
$\lbrace r \vert r\geq \frac{11}{10}\rbrace,$
imply that
\begin{align*}
\Phi_{1}(r) \geq & \frac{r^2}{r [1+\sqrt{1+\zeta_0^2}]}+\frac{r^2}{r \sqrt{1+\zeta_0^2}}+\frac{1-\sqrt{a^2+b^2}}{\zeta_0^2}\\
 =&[r-1] [\frac{1}{1+\sqrt{1+\zeta_0^2}}+\frac{1}{\sqrt{1+\zeta_0^2}}]\\
 \gtrsim & r-1\gtrsim r+1.
\end{align*}
To control $\Phi_{2}$, we have to estimate the quantity $q$ defined in \eqref{eq:defQ}: $q$ lies between $\sqrt{a^2+b^2}$ and $\sigma$. Then \eqref{eq:almostDec} and the observation that
 $\zeta_0=O(\sqrt{a^2+b^2}-1)$, which follows from \eqref{eq:zeta0}, imply that $\Phi_{2}$ is ``almost positive'', in the sense that
\begin{align*}
-\Phi_{2}\lesssim \rho_0^{\frac{1}{2}}\ll 1.
\end{align*}

This completes the proof of \eqref{eq:estTiD1} and hence of our bound on $f_{R4}$.

\section{Bound on the function $f_{R1}$ defined in \eqref{fus. f_{Rk}}}\label{sec:fR1}
By arguments essentially identical to those used in the previous appendix it is shown that it suffices to study the function
\begin{align}
\tilde{M}_1(\tau,z)=\frac{1}{i(\tau+z) q}\int_{0}^{\infty} \frac{\rho^7}{1+\rho^2} e^{-i(\tau+z) [\rho\sqrt{1+\rho^2}-q \rho]} \chi(\frac{2\rho}{ \zeta }) \ |\hat{V}(\rho)|^2\ d\rho
\end{align} which corresponds to $M_1$ in \eqref{eq:M123} in the previous appendix, recall that it is easier to study $M_2$.
Changing variables, $\rho=:\zeta r$, one finds that
\begin{align*}
\tilde{M}_{1}(\tau,z) =\frac{\zeta^8 }{i(\tau+z)q}\int_{0}^{\infty} \frac{r^7}{1+\zeta^2 r^2} e^{-i(\tau+z) \zeta^3\ X_{\zeta}(r)}\chi(2r) |\hat{V}(\zeta^2 r^2)|^2 dr,
\end{align*}
with $X_{\zeta}(r):=r \frac{\sqrt{1+(\zeta r)^2}-q}{\zeta^2}.$

To exhibit decay in $\tau$ we integrate by parts using
\begin{align}\label{eq:IntPartR1}
e^{-i(\tau+z) \zeta^3\ X_{\zeta}(r)}=\frac{1}{-i(\tau+z) \zeta^3 \partial_{r}X_{\zeta}(r)} \partial_{r}e^{-i(\tau+z) \zeta^3\ X_{\zeta}(r)}.
\end{align}
The denominator is controlled by observing that, on the support of $\chi(2r)$, which corresponds to
 $r\in [0,\frac{3}{5}]$, $X_{\zeta}$ does not have any critical points, and, using \eqref{eq:zeta}, one sees that there is a constant $C$ such that
\begin{align}
 -\partial_{r}X_{\zeta}(r)\geq C>0.
\end{align}

Thanks to the presence of the factor $r^7$ and of $\chi(2r)$ in the integrand, we can integrate by parts seven times, using \eqref{eq:IntPartR1}, without producing any boundary terms. One final integration by parts then yields
\begin{align*}
\tilde{M}_1(\tau,z)= C\frac{\zeta^8}{(\tau+z)q}[\frac{1}{[(\tau+z)\zeta^3 \partial_{r} X_{\zeta}(r)|_{r=0}]^8}+\cdots]
\end{align*}
where the terms not displayed explicitly decay more rapidly.
Simplifying the above expression one finds that
\begin{align}\label{eq:tiD1}
|\tilde{M}_1(\tau,z)|\lesssim (\tau+z)^{-9} \zeta^{-16}.
\end{align}
Together with the bound $\zeta \gtrsim \tau^{-\frac{1}{3}}$, see \eqref{eq:zetaZeta0}, this implies the desired estimate
\begin{align}
|f_{R1}(\tau,z)|,\ |\tilde{M}_1(\tau,z)|\lesssim (\tau+z)^{-\frac{11}{3}}.
\end{align}

\section{Bound on the function $f_{R3}$ defined in \eqref{fus. f_{Rk}}}\label{sec:fR3}

\begin{lemma}\label{LM:scale} For any $N\in \mathbb{N}$, there exists a $C_N>0$ such that
\begin{align}
|f_{R3}(\tau,z)|\leq C_{N} \zeta^{9} \zeta_0^2 [(\tau+z)\zeta \zeta_0^2]^{-N}.
\end{align}
\end{lemma}
Lemma \ref{LM:scale} and the bounds $\zeta,\ \zeta_0 \gtrsim \tau^{-\frac{1}{3}}$, see \eqref{eq:zetaZeta0}, obviously imply the desired estimate; (see Lemma \ref{LM:FR124}).

In order not to clutter our arguments with lengthy expressions, we consider the case where
\begin{align}
\zeta=\zeta_0=1,\ \text{hence}\ R=\frac{1}{5\sqrt{2}}.
\end{align}
(For small values of $\zeta$ and $\zeta_0$, we change variables $\rho=: \zeta r$ and
$\theta=:\zeta_0 \gamma$ and use arguments similar to those that follow below or to those used in estimating $f_{R4}$ and $f_{R1}$.)

We decompose the function $f_{R3}(\tau, z)$ into two terms corresponding to different regions of the integration variable $\theta$:
\begin{align}\label{eq:DefR3}
f_{R3}(\tau, z):=\Lambda_1(\tau,z)+\Lambda_2(\tau,z),
\end{align}
where $\Lambda_1$ is defined by
\begin{align*}
\Lambda_1(\tau,z):=&\int_{0}^{2\pi}d\alpha\int_{0}^{\pi}d\theta\int_{0}^{\infty}d\rho\text{   }\frac{\rho^8 sin\theta}{1+\rho^2} e^{-i(\tau+z)[\rho \sqrt{1+\rho^2}-\rho X(\theta,\alpha)]}
\chi_{R3}(\rho,\theta)\times
\end{align*}
\begin{align*}
\times\chi_1(\theta)
 \ g(\theta,\alpha)|V(\rho)|^2,
\end{align*}
and in the definition of the function $\Lambda_2$ one replaces $\chi_1(\theta)$ by $1-\chi_1(\theta)$.
Here $\chi_1$ is a smooth cutoff function satisfying $\chi_{1}(\theta)=1$, for
$\theta\leq \frac{3}{5}\pi$, and $\chi_{1}(\theta)=0$, for $\theta\geq \frac{3}{4}\pi$, $g$ is a polynomial in
$sine$ and $cosine$ of $\theta$ and $\alpha$, and $\chi_{R3}$ is the cutoff function defined as
$$\chi_{R3}=[1-\chi(\frac{6\theta}{5 R})] [1-\chi(\frac{2\rho}{
\zeta })]\chi(\frac{\rho}{ \zeta_0 }),$$ see \eqref{fus. f_{Rk}} and \eqref{eq:segment}.

Finally
\begin{align*}
X(\theta,\alpha):=&[\frac{\tau}{\tau+z}\sigma+\frac{z}{\tau+z}a] cos\theta+\frac{z}{\tau+z} b sin\theta cos\alpha\\
=&q cos(\theta-\beta),
\end{align*}
where
 $$q:=\sqrt{[\frac{\tau}{\tau+z}\sigma+\frac{z}{\tau+z}a]^2+[\frac{z}{\tau+z} b]^2 cos^2\alpha}>1,$$
 and $\beta=\beta(\alpha)$ is independent of $\theta$ and defined by the equations
\begin{align*}
cos\beta=\frac{\frac{\tau}{\tau+z}\sigma+\frac{z}{\tau+z}a}{q},\ sin\beta=\frac{\frac{z}{\tau+z} b cos\alpha}{q}.
\end{align*}
As shown in \eqref{eq:almostCL}, the parameter $b$ is very small, and hence
\begin{align}\label{eq:contBeta}
|\beta|\lesssim b\lesssim \rho_{0}^{\frac{1}{2}} [\sqrt{a^2+b^2}-1]\ll 1.
\end{align}

The study of $\Lambda_2$ is a little easier than that of $\Lambda_1$. We integrate by parts, using the identity
\begin{align*}
&e^{-i(\tau+z)[\rho \sqrt{1+\rho^2}-\rho X(\theta,\alpha)]}\\
=&\frac{1}{-i(\tau+z)[\sqrt{1+\rho^2}+\frac{\rho^2}{\sqrt{1+\rho^2}}-q cos(\theta-\beta)]}\partial_{\rho}[
e^{-i(\tau+z)[\rho \sqrt{1+\rho^2}-\rho X(\theta,\alpha)]}].
\end{align*}
Since $cos(\theta-\beta)<0$, for $\pi\geq\theta\geq \frac{3}{5}\pi$ and for small $\beta$, the denominator can be controlled by using
\begin{align*}
\sqrt{1+\rho^2}+\frac{\rho^2}{\sqrt{1+\rho^2}}-q cos(\theta-\beta)\geq \sqrt{1+\rho^2}.
\end{align*}
Integrating by parts as many times as one wishes, one arrives at the desired estimate on $\Lambda_2$: \begin{align}\label{eq:f2N}
|\Lambda_2|\leq C(N) (\tau+z)^{-N},
\end{align}
with $C(N)<\infty$, for any $N\in \mathbb{N}$.

We now turn to estimating $\Lambda_1$.
We write
\begin{align}
\Lambda_1(\tau,z)= \int_{0}^{2\pi}\int_{0}^{\infty}\frac{\rho^8 |V(\rho)|^2}{1+\rho^2} e^{-i(\tau+z)[\rho \sqrt{1+\rho^2 }]} \chi(\rho )[1-\chi(2 \rho)]\ \tilde{\Lambda}_1(\rho,\alpha,\tau) d\rho d\alpha,
\end{align}
where $\tilde{\Lambda}_1$ is defined by
\begin{align}
\tilde{\Lambda}_1(\rho,\alpha,\tau+z):=\int_{0}^{\pi} sin\theta\ e^{i(\tau+z) q \rho cos(\theta-\beta)}g(\theta,\alpha)[1-\chi(6 \sqrt{2} \theta)]\chi_1(\theta)\  d\theta,
\end{align}
and the cutoff function $\chi$ is from \eqref{eq:segment}.

We then integrate by parts, using
\begin{align*}
e^{iq(\tau+z)\rho cos(\theta-\beta)}=-\frac{1}{iq(\tau+z)}\frac{1}{\rho sin(\theta-\beta) }\partial_{\theta}e^{iq(\tau+z) \rho cos(\theta-\beta)}.
\end{align*}
To control the denominator we use that $\rho,\ \theta\gtrsim 1$ and $\theta \leq \frac{4}{5}\pi$ on the support of the cutoff function. With
$|\beta|\ll 1$, see \eqref{eq:contBeta}, this then implies that
\begin{align}\label{eq:ddd}
|\frac{1}{\rho sin(\theta-\beta) }|\lesssim \frac{1}{(1+\theta)(1+\rho)}.
\end{align}

Integrating by parts as many times as ones wishes one finds that, for arbitrary $N\in\mathbb{N}$, there exist finite constants $C_{N}$ such that
\begin{align}
|\Lambda_1(\tau,z)|,\ |\tilde{\Lambda}_1(\tau,z)|\lesssim C_{N}(\tau+z)^{-N}.
\end{align}
With \eqref{eq:f2N} and \eqref{eq:DefR3} this implies the desired estimate in Lemma \ref{LM:scale}.

\section{Proof of Lemma \ref{LM:FR2}}\label{sub:FR2}
As claimed in \eqref{eq:lowerBG} and proven in Appendix \ref{sec:lowerBG}, the function $|G_{a,b}|$ is strictly positive. Hence
$$[G_{a,b}(\rho,\theta,\alpha)-i0]^{-1}=G_{a,b}^{-1}(\rho,\theta,\alpha).$$

Using
\begin{align}
e^{-i\tau \sigma \rho cos\theta} sin\theta=\frac{1}{i\tau \sigma \rho} \partial_{\theta} e^{-i\tau\sigma \rho cos\theta}\label{eq:recipe10}
\end{align}
to integrate by parts one finds that
\begin{align}\label{eq:Fr2t}
F_{R2}= -\frac{1}{i\tau\sigma}[\Phi_1+\Phi_2],
\end{align}
where
\begin{align*}
\Phi_1(\tau):=&  \int_{0}^{2\pi}\int_{0}^{\infty} e^{-i\tau Y_{\sigma}(\rho,0)} G_{a,b}^{-2}(\rho,0,\alpha) \frac{\rho^5}{1+\rho^2} g(0,\alpha)|\hat{V}(\rho)|^2 \chi_{R2}(0,\frac{2\rho}{\zeta},\frac{\rho}{\zeta_0})\ d\rho d\alpha\\
= & C\int_{0}^{\infty} e^{-i\tau \rho [\sqrt{1+\rho^2}-\sigma]} \frac{1}{[\sqrt{1+\rho^2}-a]^2} \frac{\rho^5}{1+\rho^2} |\hat{V}(\rho)|^2[1-\chi(\frac{2\rho}{\zeta})]\chi(\frac{\rho}{\zeta_0})\ d\rho,\\
\Phi_2(\tau):= & \int_{0}^{2\pi}\int_{0}^{\pi}\int_{0}^{\infty} e^{-i\tau Y_{\sigma}(\rho,\theta)}  \frac{\rho^5}{1+\rho^2} |\hat{V}(\rho)|^2 \ \partial_{\theta}[G_{a,b}^{-2}(\rho,\theta,\alpha)g(\theta,\alpha)\chi_{R2}(\frac{\theta}{R},\frac{2\rho}{\zeta},\frac{\rho}{\zeta_0})]\ d\rho d\theta d\alpha
\end{align*} and $\chi_{R2}$ is the cutoff function defined as
$$\chi_{R2}(\frac{\theta}{R},\frac{2\rho}{\zeta},\frac{\rho}{\zeta_0}):=\chi(\frac{6\theta}{5 R}) [1-\chi(\frac{2\rho}{ \zeta })]\chi(\frac{\rho}{ \zeta_0 }),$$ see \eqref{eq:segment}.
In the second equation for $\Phi_1$, $C$ is a constant and the variable $\alpha$ is integrated out.

The dominant contribution to $F_{R2}$ is the one proportional to $\Phi_1$. As noted in \eqref{eq:critical}, the phase $\rho[\sqrt{1+\rho^2}-\sigma]$ has a non-degenerate critical point, $\rho=\zeta$, in the integration region considered here. In order to scrutinize the neighborhood of $\rho=\zeta$ and exploit the smallness of $\zeta$, we introduce a new variable, $r$, by setting
\begin{align}\label{eq:rhor}
r:=\zeta^{-1}(\rho-\zeta).
\end{align}
Then
$$e^{-i\tau Y(\rho,0)}=e^{-i\tau a_0(\zeta)} e^{-is \ a_2(\zeta,\zeta r)\ r^2 },$$
where $s:=\zeta^3 \tau$, and
$a_0$ and $a_2$ appear in the phase
\begin{align}
\rho[\sqrt{1+\rho^2}-\sigma]= a_0(\zeta)+ a_2(\zeta,\zeta r) \zeta^3 r^2,
\end{align}
with $a_0(\zeta)\in \mathbb{R}$ a constant independent of $r$ and $a_{2}$ a smooth real-valued function of $r$. To see that the critical point at $r=0$ is non-degenerate we note that
$$a_2(\zeta,\zeta r)|_{r=0}=\frac{3}{2}\frac{1}{\sqrt{1+\zeta^2}}-\frac{1}{2}\frac{\zeta^2}{(1+\zeta^2)^{\frac{3}{2}}}\geq \frac{1}{\sqrt{1+\zeta^2}}.$$

Concerning the other factors in the integrand appearing in the definition of $\Phi_1$ we note that the function $\frac{\zeta^4}{[\sqrt{1+\rho^2}-a]^2}|_{\rho=\zeta r+\zeta}$ is uniformly smooth in $r$. This is seen by recalling that $\frac{1}{\sqrt{1+\rho^2}-a}=G_{a,b}^{-1}(\rho,0,0)$ and then using \eqref{eq:lowerBG}.

The function $\Phi_{1}$ can now be written in the form
\begin{align}\label{eq:reF1}
\Phi_{1}=C e^{-i\tau a_0(\zeta)} \zeta^2\int_{-\infty}^{\infty} e^{-is \ a_2(\zeta,\zeta r)\ r^2 } H(r)\ dr,
\end{align}
where $C$ is a constant and $H$ is a smooth function of compact support.
Applying the standard stationary phase method we find that
\begin{align}\label{eq:estF1}
|\Phi_1(\tau)|\lesssim \zeta^2  s^{-\frac{1}{2}}= \zeta^{\frac{1}{2}} \tau^{-\frac{1}{2}}\lesssim \tau^{-\frac{1}{2}}.
\end{align}
The details are standard but tedious and are omitted.

Turning to $\Phi_2$, we claim that a rather crude analysis will yield the desired decay
$\propto \tau^{-\frac{1}{2}}$.
To avoid unnecessarily complicated formulae we only consider the case where
\begin{align}
\zeta=\zeta_0=1,\ \text{ hence}\ R=\frac{1}{5\sqrt{2}}.
\end{align}
(For small $\zeta,\ \zeta_0$, we change variables: $\rho=:\zeta r$ and $\theta=:\zeta_0 \beta$. This will lead to the desired estimate.)

We first perform the $\theta$-integral and obtain that
\begin{align}\label{eq:newF2}
\Phi_2(\tau)=\int_{0}^{2\pi}\int_{0}^{\infty} e^{-i\tau \rho[\sqrt{1+\rho}-\sigma]}\tilde{\Phi}_2 (\rho,\alpha,\tau) [1-\chi(2\rho)] \chi(\rho) d\rho d\alpha,
\end{align}
where $\tilde{\Phi}_2$ is given by
$$\tilde{\Phi}_2(\rho,\alpha,\tau):=\int_{0}^{\infty}e^{-i\tau\sigma \rho [1-cos(\theta)]} H(\rho,\theta,\alpha)\chi(6 \sqrt{2}\theta)\ d\theta,$$
and $H(\rho,\theta,\alpha)$ is a uniformly smooth function.

Due to the presence of the cutoff functions, we only need to consider the region
$$\rho\geq \frac{3}{5}\ \text{and}\ \theta\in [0,\frac{1}{6}].$$

The completion of the argument is standard: We observe that $1-cos(\theta)=\frac{1}{2}\theta^2[1+O(\theta^2)]$
in the domain of small $\theta$-values appearing in the integral defining $\tilde{\Phi}_2$ and that $\rho$
can be replaced by a constant, because it is bounded away from zero. These observations enable us to apply standard stationary phase arguments that show that
\begin{align}
|\tilde{\Phi}_2(\tau)|\lesssim \tau^{-\frac{1}{2}}.
\end{align}
Plugging this bound into \eqref{eq:newF2} we obtain that
\begin{align}
|\Phi_2(\tau)|\lesssim \tau^{-\frac{1}{2}}.
\end{align}

With the bound on $\Phi_1$ in \eqref{eq:estF1} and the decomposition of $F_{R2}$ in \eqref{eq:Fr2t} this clearly yields the desired decay estimate on $F_{R2}$.

\def\cprime{$'$} \def\cprime{$'$} \def\cprime{$'$} \def\cprime{$'$}
  \def\cprime{$'$} \def\cprime{$'$}

\end{document}